\documentclass[11pt]{article}
\usepackage{amsmath,amssymb,amsthm,mathtools}
\usepackage[usenames,dvipsnames]{xcolor}
\usepackage{algorithm,algpseudocode,Levi}
\usepackage{paralist}

\usepackage[capitalize]{cleveref}

\newtheorem{theorem}{Theorem}[section]
\newtheorem{lemma}[theorem]{Lemma}

\newtheorem{proposition}[theorem]{Proposition}

\newtheorem{corollary}[theorem]{Corollary}
\newtheorem{observation}[theorem]{Observation}
\newtheorem{definition}[theorem]{Definition}

\def\FullBox{\hbox{\vrule width 8pt height 8pt depth 0pt}}
\newcommand{\QED}{\;\;\;\FullBox}
\renewenvironment{proof}{\noindent{\bf Proof:~}}{\hfill\QED}
\newenvironment{proofof}[1]{\noindent{\bf Proof of {#1}:~}}{\hfill\(\QED\)}

\makeatletter
\newtheorem*{rep@theorem}{\rep@title}
\newcommand{\newreptheorem}[2]{%
	\newenvironment{rep#1}[1]{%
		\def\rep@title{#2 \ref{##1}}%
		\begin{rep@theorem}}%
		{\end{rep@theorem}}}
\makeatother
\newreptheorem{theorem}{Theorem}

\DeclarePairedDelimiter\qu{\langle\!\langle}{\rangle\!\rangle}

\newcommand{\PCP}{\mathsf{PCP}}
\newcommand{\PCPP}{\mathsf{PCPP}}
\newcommand{\PCU}{\mathsf{PCU}}

\newcommand{\Value}{\mathsf{\bf Value}}
\newcommand{\PCUSS}{\mathsf{PCUSS}}
\newcommand{\Proof}{\mathrm{{\bf Proof}}}

\newcommand{\Spiel}{\mathsf{{ Spiel}}}
\newcommand{\Dom}{\mathrm{{ Dom}}}
\newcommand{\etal}{\textit{et al.}}
\renewcommand{\L}{\mathcal{L}}

\renewcommand{\polylog}{\mathrm{{polylog}\;}}
\newcommand{\polylogell}{\mathrm{polylog}^{(\ell)}}

\def\authornameAL{Amit Levi}
\def\authoraffiAL{University of Waterloo, Canada. Email: \href{mailto: amit.levi@uwaterloo.ca}{amit.levi@uwaterloo.ca}. Research supported by the David R. Cheriton Graduate Scholarship. Part of this work was done while the author was visiting the Technion.}
\def\authornameEF{Eldar Fischer}
\def\authoraffiEF{Technion - Israel Institute of Technology, Israel. Email: \href{mailto:eldar@cs.technion.ac.il}{eldar@cs.technion.ac.il}.}
\def\authornameOBE{Omri Ben-Eliezer}
\def\authoraffiOBE{Tel Aviv University, Israel. Email: \href{mailto: omrib@mail.tau.ac.il}{omrib@mail.tau.ac.il}.}
\def\authornameRR{Ron D.~Rothblum}

\def\authoraffiRR{Technion - Israel Institute of Technology,
  Israel. Email:
  \href{mailto:rothblum@cs.technion.ac.il}{rothblum@cs.technion.ac.il}. Supported
  in part by the Israeli Science Foundation (Grant No. 1262/18), a
  Milgrom family grant and the Technion Hiroshi Fujiwara cyber
  security research center and the Israel cyber directorate.}

\begin{document}

\title{Hard properties with (very) short PCPPs and their applications}
\author{
	\authornameOBE\thanks{\authoraffiOBE}
	\and
	\authornameEF\thanks{\authoraffiEF}
	\and
	\authornameAL\thanks{\authoraffiAL}
	\and
	\authornameRR\thanks{\authoraffiRR}
}
\date{}
\maketitle
\begin{abstract}

We show that there exist properties that are maximally hard for
testing, while still admitting $\PCPP$s with a proof size very close to linear. Specifically, for every fixed $\ell$, we construct a
property $\calP^{(\ell)}\subseteq\zo^n$ satisfying the following: Any testing algorithm
for $\calP^{(\ell)}$ requires $\Omega(n)$ many queries, and yet
$\calP^{(\ell)}$ has a constant query $\PCPP$ whose proof size is
$O(n\cdot \log^{(\ell)}n)$, where $\log^{(\ell)}$ denotes the $\ell$
times iterated log function (e.g., $\log^{(2)}n = \log \log n$). The
best previously known upper bound on the $\PCPP$ proof size for a
maximally hard to test property was $O(n \cdot \polylog{n})$.

As an immediate application, we obtain stronger separations between
the standard testing model and both the tolerant testing model and
the erasure-resilient testing model: for every fixed
$\ell$, we construct a property that has a constant-query tester,
but requires $\Omega(n/\log^{(\ell)}(n))$ queries for every tolerant
or erasure-resilient tester.

\end{abstract}

\section{Introduction}

Probabilistically checkable proofs ($\PCP$s) are one of the landmark
achievements in theoretical computer science. Loosely speaking,
$\PCP$s are proofs that can be verified by reading only a very small
(i.e., constant) number of bits. Beyond the construction of highly
efficient proof systems, $\PCP$s have myriad applications, most
notably within the field of hardness of approximation.

A closely related variant of $\PCP$s, called \emph{probabilistically
  checkable proofs of proximity} ($\PCPP$s), was introduced
independently by Ben-Sasson~\etal{}~\cite{BGHSV06} and Dinur and
Reingold~\cite{DR06}. In the $\PCPP$ setting, a verifier is given oracle access
to both an input $x$ and a proof $\pi$. It should make a few (e.g.,
constant) number of queries to both oracles to ascertain whether
$x \in \L$. Since the verifier can only read a few of the input bits,
we only require that it rejects inputs that are \emph{far} (in Hamming
distance) from $\L$, no matter what proof $\pi$ is provided. $\PCPP$s
are highly instrumental in the construction of standard
$\PCP$s. Indeed, using modern terminology, both the original algebraic
construction of $\PCP$s \cite{ALMSS98} (see also \cite{BGHSV06}) as
well as Dinur's \cite{Din07} combinatorial proof utilize $\PCPP$s.

By combining the seminal works of Ben-Sasson and Sudan~\cite{BS08} and
Dinur~\cite{Din07}, one can obtain $\PCP$s and $\PCPP$s with only
poly-logarithmic (multiplicative) overhead. More specifically, the
usual benchmark for $\PCPP$s is with respect to the
$\mathtt{CircuitEval}$ problem, in which the verifier is given
explicit access to a circuit $C$ and oracle access to both an input
$x$ and a proof $\pi$, and needs to verify that $x$ is close to the
set $\{ x' : C(x')=1\}$. The works of \cite{BS08,Din07} yield a
$\PCPP$ whose length is quasilinear in the size $|C|$ of the circuit
$C$.\footnote{Note that a $\PCPP$ for $\mathtt{CircuitEval}$ can be
  easily used to construct a $\PCP$ for $\mathtt{CircuitSAT}$ with
  similar overhead (see \cite[Proposition 2.4]{BGHSV06}).}

Given the important connections both to constructions of efficient
proof-systems, and to hardness of approximation, a central question in
the area is whether this result can be improved: Do $\PCPP$s with only
a \emph{constant} overhead exist? In a recent work,
Ben~Sasson~\etal{}~\cite{BKKMS16} construct $\PCP$s with
constant overhead, albeit with very large query complexity (as well as
a non-uniform verification procedure).\footnote{Although it is not
  stated in \cite{BKKMS16}, we believe that their techniques can also
  yield $\PCPP$s with similar parameters.} To verify that $C(x)=1$
the verifier needs to make $|C|^{\delta}$ queries, where $\delta>0$
can be any fixed constant.

Given the lack of success (despite the significant interest) in constructing
constant-query $\PCPP$s with constant overhead, it may be the case that there exist languages that do not have such efficient
$\PCPP$s. A natural class of candidate languages for which such
$\PCPP$s may not exist are languages for which it is \emph{maximally}
hard to test whether $x \in \L$ or is far from such, \emph{without} a
$\PCPP$ proof. In other words, languages (or rather properties) that
do not admit sub-linear query testers. Thus, we investigate the following question:

\begin{quote}
  \emph{Supposing that $\L$ requires $\Omega(n)$ queries for every (property)
    tester, must any constant-query $\PCPP$ for $\L$ have proof
    length $n \cdot (\log{n})^{\Omega(1)}$?}
\end{quote}

\subsection{Our Results}
Our first main result answers the above question negatively, by
constructing a property that is maximally hard for testing, while
admitting a very short $\PCPP$. For the exact theorem statement, we
let $\log^{(\ell)}$ denote the $\ell$ times iterated $\log$
function. That is, $\log^{(\ell)}(n) = \log(\log^{(\ell-1)}(n))$ for
$\ell \geq 1$ and $\log^{(0)}n=n$.

\begin{theorem}[informal restatement of Theorem~\ref{thm:PCUSS-construction}]\label{informalthm:main} For every constant integer $\ell\in\N$, there exists a property $\calP\subseteq \{0,1\}^n$ such that any
  testing algorithm for $\calP$ requires $\Omega(n)$ many queries, while
  $\calP$ admits a (constant query\footnote{For detection radius (or proximity parameter) $\eps > 0$ and constant soundness, the particular query complexity of the $\PCPP$ system is bounded by $(2^{\ell} / \eps)^{O(\ell)}$.}) $\PCPP$ system with proof length
  $O(n\cdot \log^{(\ell)}(n))$.
\end{theorem}

We remark that
all such maximally hard properties cannot have constant-query $\PCPP$
proof-systems with a \emph{sub-linear} length proof
string (see Proposition~\ref{prop:NoSublinearPCPP-HardProperties}), leaving only a small gap of $\log^{(\ell)}(n)$ on the proof
length in \cref{informalthm:main}.
\medskip

Beyond demonstrating that $\PCPP$s with extremely short proofs exist
for some hard properties, we use \cref{informalthm:main} to derive
several applications. We proceed
to describe these applications next.

\paragraph{Tolerant Testing.}
Recall that property testing (very much like $\PCPP$s) deals with
solving approximate decision problems. A tester for a property $\calP$
is an algorithm that given a sublinear number of queries to its input
$x$, should accept (with high probability) if $x \in \calP$ and reject
if $x$ is \emph{far} from $\calP$ (where, unlike $\PCPP$s, the tester is not provided with any proof).

The standard setting of property testing is arguably fragile,
since the testing algorithm is only guaranteed to accept all functions
that exactly satisfy the property. In various settings and
applications, accepting only inputs that exactly have a certain
property is too restrictive, and it is more beneficial to distinguish
between inputs that are close to having the property, and those that
are far from it. To address this question, Parnas, Ron and
Rubinfeld~\cite{parnas2006tolerant} introduced a natural
generalization of property testing, in which the algorithm is required
to accept functions that are close to the property. Namely, for parameters $0\le\epsilon_0<\epsilon_1\le1$, an
\emph{$(\eps_0,\eps_1)$-tolerant testing algorithm} is given an oracle access to the
input, and is required
to determine (with high probability) whether a given input is
$\epsilon_0$-close to the property or whether it is $\epsilon_1$-far
from it. As observed in~\cite{parnas2006tolerant}, any standard
testing algorithm whose queries are uniformly (but not necessarily
independently) distributed, is inherently tolerant to some
extent. Nevertheless, for many problems, strengthening the tolerance
requires applying advanced methods and devising new algorithms (see
e.g.,~\cite{FN07,KS09,CGR13,BMR16,BCELR18}).

It is natural to ask whether tolerant testing is strictly harder than standard testing. This question was
explicitly studied by Fischer and Fortnow~\cite{FF06}, who used
$\PCPP$s with polynomial size proofs to show that there exists a
property $\calP \subseteq \{0,1\}^n$ that admits a tester with constant query
complexity, but such that every tolerant tester for $\calP$
has query complexity $\Omega(n^c)$ for some $0<c<1$. Using modern quasilinear $\PCPP$s~\cite{BS08, Din07} in combination with the techniques of~\cite{FF06} it is possible to construct a property demonstrating a better separation, of constant query complexity for standard testing versus $\Omega(n / \polylog n)$ for tolerant testing.

Using \cref{informalthm:main} we can obtain an improved separation
between testing and tolerant testing:

{\begin{theorem}[informal restatement of Theorem~\ref{thm:Tol-Sep-Notintro}]\label{thm:Tol-Sep-Intro}
  For any constant integer $\ell\in \N$, there exist a property of
  boolean strings $\calP\subseteq\{0,1\}^n$ and a constant $\eps_1\in (0,1)$ such that
  $\calP$ is $\eps_0$-testable for any $\eps_0>0$ with a number of queries\footnote{The (constant) query complexity of the intolerant $\epsilon$-tester has the same asymptotic bound as in Theorem \ref{informalthm:main}.} independent of
  $n$, but for any $\eps_0\in(0,\eps_1)$, every $(\eps_0,\eps_1)$-tolerant tester for $\calP$ requires $\Omega(n/ \polylogell n)$ many queries.
\end{theorem}}

\paragraph{Erasure-Resilient Testing.} Another variant of the property
testing model is the \emph{erasure-resilient testing} model. This
model was defined by Dixit et.~al.~\cite{DRTV18} to address cases
where data cannot be accessed at some domain points due to privacy
concerns, or when some of the values were adversarially erased. More
precisely, an $\alpha$-erasure-resilient $\eps$-tester gets as input parameters
$\alpha,\epsilon\in(0,1)$, as well as oracle access to a function $f$,
such that at most an $\alpha$ fraction of its values have been
erased. The tester has to accept with high probability if there is a
way to assign values to the erased points of $f$ such that the resulting
function satisfies the desired property. The tester has to reject with
high probability if for every assignment of values to the erased
points, the resulting function is still $\eps$-far from the desired
property.

Similarly to the tolerant testing scenario, $\PCPP$s were also used in \cite{DRTV18} to
show that there exists a property of boolean strings of length $n$
that has a tester with query complexity independent of $n$, but for any constant
$\alpha>0$, every $\alpha$-erasure-resilient
tester is required to query $\Omega(n^c)$ many bits for some $c>0$,
thereby establishing a separation between the models. Later,
in~\cite{RRV19} $\PCPP$ constructions were used to provide a
separation between the erasure-resilient testing model and the
tolerant testing model.

Similarly to the tolerant testing case, we use Theorem~\ref{informalthm:main} to prove a stronger separation between the erasure-resilient testing model and the standard testing model.

\medskip


{\begin{theorem}[informal restatement of Theorem~\ref{thm:ER-Sep-notintro}]\label{thm:ER-Sep-Intro}
  For any constant integer $\ell\in\N$, there exist a property of
  boolean strings $\calP\subseteq\{0,1\}^n$ and a constant $\eps_1\in(0,1)$ such that
  $\calP$ is $\eps$-testable for any $\eps>0$ with number of queries\footnote{Again, the query complexity of the (non erasure resilient) $\epsilon$-tester has the same asymptotic bound as in Theorem \ref{informalthm:main}.} independent of
  $n$, but for any $\alpha=\Omega(1/\log^{(\ell)}n)$ and $\eps\in (0,\eps_1)$ such that $\eps<1-\alpha$, any $\alpha$-erasure-resilient $\eps$-tester is required to query
  $\Omega(n/\polylogell n)$ many bits.
\end{theorem}}

{\paragraph{Secret Sharing applications.} 
	As an additional application of our techniques we also obtain a new
  type of \emph{secret sharing scheme}. Recall that in a secret
  sharing scheme, a secret value $b \in \{0,1\}$ is shared between $n$
  parties in such a way that only an authorized subset of the users
  can recover the secret. We construct a secret sharing scheme in
  which no subset of $o(n)$ parties can recover the secret and yet
  it is possible for each one of the parties to recover the secret, if
  given access to a $\PCPP$-like proof, with the guarantee that no
  matter what proof-string is given, most parties will either recover
  $b$ or reject.

  We obtain such a secret sharing scheme through a notion called
  \emph{Probabilistically Checkable Unveiling of a Shared Secret
    ($\PCUSS$)}, which will be central in our work. This notion is
  loosely described in Subsection~\ref{subsec:Techniques} and formally
  defined in Section~\ref{sec:PCUandPCUSS}.  }

\subsection{Techniques}\label{subsec:Techniques}
Central to our construction are {(univariate)} polynomials over a
{finite} field $\F$. A basic fact is that a random polynomial
$p : \F \to \F$ of degree (say) $|\F|/2$, evaluated at any set of at
most $|\F|/2$ points, looks exactly the same as a totally random
function $f : \F \to \F$. This is despite the fact that a random
function is very far (in Hamming distance) from the set of low degree
polynomials. Indeed, this is the basic fact utilized by Shamir's
secret sharing scheme \cite{shamir1979share}.

Thus, the property of being a degree-$|\F|/2$ univariate polynomial is a hard problem
to decide for any tester, in the sense that such a tester must make
$\Omega(|\F|)$ queries to the truth table of the function in order to
decide. Given that, it seems natural to start with this property in
order to prove \cref{informalthm:main}. Here we run into two
difficulties. First, the property of being a low degree polynomial is defined over a large alphabet,
whereas we seek a property over boolean strings. Second, the best
known $\PCPP$s for this property have quasi-linear length \cite{BS08},
which falls short of our goal.

To cope with these difficulties, our approach is to use composition, or
more accurately, an iterated construction. The main {technical}
contribution of this paper lies in the mechanism enabling this
iteration. More specifically, rather than having the property contain
the explicit truth table of the low degree polynomial $p$, we would
like to use a more redundant representation for encoding each value
$p(\alpha)$. This encoding should have several properties:

\begin{itemize}
\item It must be the case that one needs to read (almost) the entire
  encoding to be able to decode $p(\alpha)$. This feature of the
  encoding, which we view as a secret-sharing type of property, lets
  us obtain a hard to test property over boolean strings.

\item The encoding need not be efficient, and in fact it will be made long enough to eventually subsume the typical length of a $\PCPP$ proof-string for the low degree property, when calculated with respect to an {\em unencoded} input string.

\item Last but not least, we need the value to be decodable using very few queries, when given access to an auxiliary $\PCP$-like proof string. This would allow us to ``propagate'' the $\PCPP$ verification of the property across iterations.

\end{itemize}

In more detail, we would like to devise a (randomized) encoding of
strings in $\{0,1\}^k$ by strings in $\{0,1\}^m$. The third requirement
listed above can be interpreted as saying that given oracle access to
$v\in\{0,1\}^m$ and explicit access to a value $w\in\{0,1\}^k$, it
will be possible verify that $v$ indeed encodes $w$ using a
$\PCPP$-like scheme, i.e. by providing a proof that can be verified
with a constant number of queries. We
refer to this property as a \emph{probabilistically checkable
  unveiling ($\PCU$)}\footnote{In fact, we will use a stronger variant where the access to $w$ is also restricted.}. Note that in our setting a single
value $w$ may (and usually will) have more than one valid encoding.

Going back to the first requirement of the encoding, we demand that without a
proof, one must query at least $\Theta(m)$ bits of $v$ to obtain
\emph{any} information about the encoded $w$, or even discern that $v$
is indeed a valid encoding of some value. Given this combination of requirements, we refer to the verification procedure as a
\emph{Probabilistically Checkable Unveiling of a Shared Secret
  ($\PCUSS$)}.


Low degree polynomials can be used to obtain a $\PCUSS$ based on
Shamir's secret sharing scheme. More specifically, to encode a $k$ bit
string $w$, we take a random polynomial whose values on a subset
$H \subseteq \F$ are exactly equal to the bits of $w$. However, we
provide the values of this polynomial only over the sub domain
$\F\setminus H$. Then, the encoded value is represented by the
(interpolated) values of $g$ over $H$, which admit a $\PCU$ scheme. On
the other hand, the ``large independence'' feature of polynomials
makes the encoded value indiscernible without a a supplied proof
string, unless too many of the values of $g$ over $\F\setminus H$ are
read, thus allowing for a $\PCUSS$.

This construction can now be improved via iteration. Rather than
explicitly providing the values of the polynomial, they will be
provided by a $\PCUSS$ scheme. Note that the $\PCUSS$ scheme that we
now need is for strings of a (roughly) exponentially smaller
size. The high level idea is to iterate this construction $\ell$ times
to obtain the $\ell$ iterated log function in our theorems.

At the end of the recursion, i.e., for the smallest blocks at the
bottom, we utilize a linear-code featuring both high distance and high
dual distance, for a polynomial size $\PCUSS$ of the encoded
value. This is the only ``non-constructive'' part in our construction,
but since the relevant block size will eventually be less than
$\log\log(n)$, the constructed property will still be uniform with
polynomial calculation time (the exponential time in
$\poly(\log\log(n))$, needed to construct the linear-code matrix,
becomes negligible).

Our $\PCUSS$ in particular provides a property that is hard to test (due 
to its shared secret feature), and yet has a near-linear $\PCPP$ through 
its unveiling, thereby establishing Theorem~\ref{informalthm:main}. We utilize this property for separation results in a 
similar manner to \cite{FF06} and \cite{DRTV18}, by considering a 
weighted version of a ``$\PCPP$ with proof'' property, where the proof part 
holds only a small portion of the total weight. The $\PCPP$ proof part enables a constant query test, whereas if the $\PCPP$ proof is deleted, efficient testing is no longer possible.

\subsection{Related work}

\paragraph{Short $\PCPP$s.}
For properties which can be verified using a circuit of size $n$,
\cite{BGHSV06} gave $\PCPP$ constructions with proof length
$n\cdot \exp(\poly(\log\log n))$ and with a query complexity of
$\poly(\log\log n)$, as well as slightly longer proofs with constant
query complexity. Later, Ben-Sasson and Sudan~\cite{BS08} gave
constructions with quasilinear size proofs, but with slightly higher
query complexity. The state of the art construction is due to
Dinur~\cite{Din07} who, building on \cite{BS08}, showed a $\PCPP$
construction with proof length that is quasilinear in the circuit size
and with constant query complexity. In a recent work
Ben~Sasson~\etal{}~\cite{BCGRS17} constructed an \emph{interactive}
version of $\PCPP$s \cite{BCS16,RRR16} of strictly linear length and
constant query complexity.

\paragraph{Tolerant Testing.}
The tolerant testing framework has received significant attention in the past decade. Property testing of dense graphs, initiated by \cite{GGR98}, is inherently tolerant by the canonical tests of Goldreich and Trevisan~\cite{GT03}.
 Later, Fischer and Newman \cite{FN07} (see also \cite{BF_CCC18}) showed that every testable (dense) graph property admits a tolerant testing algorithm for \emph{every} $0<\epsilon_0<\epsilon_1<1$, which implies that $O(1)$ query complexity testability is equivalent to distance approximation in the dense graph model.
Some properties of boolean functions were also studied recently in the tolerant testing setting. In particular, the properties of being a \emph{$k$-junta} (i.e. a function that depends on $k$ variables) and being \emph{unate} (i.e., a function where each direction is either monotone increasing or monotone decreasing)~\cite{BCELR18,LW19,DMN19}.

\paragraph{Erasure-resilient Testing.}
For the erasure resilient model, in addition to the separation between that model and the standard testing model, \cite{DRTV18} designed efficient erasure-resilient testers for important properties, such as monotonicity and convexity. Shortly after, in~\cite{RRV19} a separation between the erasure-resilient testing model and the tolerant testing model was established. The last separation requires an additional construction (outside $\PCPP$s), which remains an obstacle to obtaining better than polynomial separations.

\section{Preliminaries}
We start with some notation and central definitions. For a set $A$, we let $2^A$ denote the power-set of $A$. For two strings
$x,y\in\{0,1\}^*$ we use $x\sqcup y$ to denote string
concatenation. 

For an integer $k$, a field $\F=\mathrm{GF}(2^k)$ and $\alpha\in\F$,
we let $\qu{\alpha}\in\{0,1\}^k$ denote the binary representation of
$\alpha$ in some canonical way.  

\noindent For two sets of strings $A$ and
$B$ 
we use $A\sqcup B$ to denote the set
$\{a\sqcup b\;\mid \; a\in A,\; b\in B \}$. For a collection of sets
$\{A(d)\}_{d\in D}$ we use $\bigsqcup_{d\in D}A(d)$ to denote the set
of all possible concatenations $\bigsqcup_{d\in D}a_d$, where $a_d\in A(d)$
for every $d \in D$. 

Throughout this paper we use boldface letters to denote random variables, and assume a fixed canonical ordering over the elements in all the sets we define.  {For a set $D$, we write $\bv\sim D$ to denote a random variable resulting from a uniformly random choice of an element $v\in D$.}

\subsection{Error correcting codes and polynomials over finite fields}

The relative Hamming distance of two strings $x,y \in \Sigma^n$ is
defined as
$\dist(x,y) = \frac{1}{n} \cdot |\{i\in [n]\mid x_i\neq y_i \}|$. For a
string $x \in \Sigma^n$ and a non-empty set $S \subseteq \Sigma^n$, we
define $\dist(x,S) = \min_{y \in S} \dist(x,y)$. The following plays
a central role in many complexity-related works, including ours.

\begin{definition}
 A \emph{code} is an injective function
$C : \Sigma^k \to \Sigma^n$. If $\Sigma$ is a finite field and $C$ is
a linear function (over $\Sigma$), then we say that $C$ is a
\emph{linear code}.  The \emph{rate} of $C$ is defined as $k/n$, whereas the
\emph{minimum relative distance} is defined as the minimum over all
distinct $x,y \in \Sigma^k$ of $\dist(C(x),C(y))$.

An \emph{$\epsilon$-distance code} is a code whose minimum relative distance
is at least $\epsilon$. When for a fixed $\epsilon>0$ we have a family of
$\epsilon$-distance codes (for different values of $k$), we refer to its members
as \emph{error correcting codes}.
\end{definition}

In this work we use the fact that efficient codes with constant
rate and constant relative distance exist.  Moreover,
there exist such codes in which membership can be decided by a
quasi-linear size Boolean circuit.

\begin{theorem}[see e.g., \cite{Spi96}]\label{thm:spiel}
  There exists a linear code $\Spiel:\zo^k \to \{0,1\}^{100 k}$ with constant relative distance, for which membership
  can be decided by a $k \cdot \polylog{k}$ size Boolean circuit.
\end{theorem}

Actually, the rate of the code in \cite{Spi96} is significantly
better, but since we do not try to optimize constants, we
use the constant $100$ solely for readability. In addition, the code described
in~\cite{Spi96} is  \emph{linear time decodeable}, but we do not
make use of this feature throughout this work.

\medskip We slightly abuse notation, and for a finite field $\F$ of
size $2^k$, view the encoding given in \cref{thm:spiel} as
$\Spiel:\F \to \{0,1\}^{100k}$, by associating $\{0,1\}^k$ with $\F$ in
the natural way. Note that for $f:\F \to \F$, it holds that
$\qu {f(\beta)}\in\{0,1\}^{k}$ for every $\beta \in \F$, and therefore
$\Spiel(f(\beta))\in\{0,1\}^{100k}$. We slightly abuse
notation, and for a function $f: \F \to \F$ we write $\Spiel(f)$ to
denote the length $100k \cdot 2^k$ bit string
$\bigsqcup_{\beta\in \F}\Spiel(f(\beta))$ (where we use the canonical
ordering over $\F$). 


\begin{definition} Let $\calC_\F$ denote the set of polynomials
	$g:\F\to \F$ such that $\deg({g})\le \frac{|\F |}{2}$.
\end{definition}
The following lemma of
\cite{horowitz1972fast}, providing a fast univariate interpolation,  will be an important tool in this work.

\begin{lemma}[\cite{horowitz1972fast}]\label{lem:FastInterpolation}
	Given a set of pairs $\{(x_1,y_1),\ldots,(x_r,y_r)\}$ with all $x_i$
	distinct, we can output the coefficients of $p(x)\in \F[X]$ of
	degree at most $r-1$ satisfying $p(x_i)=y_i$ for all $i\in[r]$, in
	$O(r \cdot \log^3(r))$ additions and multiplications in
	$\F$.

\end{lemma}
%

\medskip
\noindent The next lemma states that a randomly chosen function $\blambda :\F \to \F$ is far from any low degree polynomial with very high probability.
\begin{lemma} \label{lem:rand-function-far-F}With probability at least $1-o(1)$, a uniformly random function $\blambda:\F \to\F$ is $1/3$-far from $\calC_{\F}$.
\end{lemma}
\begin{proof}Consider the size of a ball of relative radius $1/3$ around some function $\lambda:\F \to \F$ in the space of functions from $\F$ to itself. The number of points (i.e., functions from $\F\to\F$) contained in this ball is at most \[\binom{|\F |}{|\F |/3}\cdot |\F |^{|\F|/3}\le (3e|\F |)^{|\F |/3}.\]
	{By the fact that the size of $\calC_{\F}$ is $|\F|^{|\F|/2+1}$}, the size of the set of points that are at relative
	distance {at most} $1/3$ from any point in $\calC_\F$ is at most \[|\F
	|^{|\F |/2+1}\cdot (3e|\F |)^{|\F |/3}= o({|\F|^{|\F|}}).\]
	{The lemma follows by observing that there are
		$|\F|^{|\F|}$ functions from $\F$ to itself.}
\end{proof}

\medskip

\subsubsection{Dual distance of linear codes}

We focus here specifically on a linear code $C:\F^k\to\F^n$, 
and consider the linear subspace of its image,
$V_C=\{C(x):x\in\F^k\}\subseteq\F^n$. We define the {\em distance} of a linear space
as $\dist(V)=\min_{v\in V\setminus \{0^n\}}\dist(v,0^n)$, and note that in the case of $V$ being
the image $V_C$ of a code $C$, this is identical to $\dist(C)$. For a linear code, it helps to
investigate also \emph{dual distances}.

\begin{definition}
Given two vectors $u,v\in\F^n$, we define their \emph{scalar product} as
$u\cdot v=\sum_{i\in [n]}u_iv_i$, where multiplication and addition are
calculated in the field $\F$.
Given a linear space $V\subseteq\F^n$, its \emph{dual space} is the
linear space $V^{\bot}=\{u:\forall v\in V, u\cdot v=0\}$. In other
words, it is the space of vectors who are orthogonal to all members of
$V$.The \emph{dual distance} of the space $V$ is simply defined as $\dist(V^{\bot})$.
\end{definition}

For a code $C$, we define its \emph{dual distance}, $\dist^{\bot}(C)$, as
the dual distance of its image $V_C$. We call $C$ an \emph{$\eta$-dual-distance} code
if $\dist^{\bot}(C)\geq\eta$. The following well-known lemma is essential to us,
as it will relate to the ``secret-sharing'' property that we define later.

\begin{lemma}[See e.g., {\cite[Chapter $1$, Theorem $10$]{macwilliams1977theory}}]\label{lem:dual-to-secret}
  Suppose that $C:\F^k\to\F^n$ is a linear $\eta$-dual distance code,
  let $Q\subset [n]$ be any set of size less than $\eta \cdot n$, and
  consider the following random process for picking a function
  $\bu:Q\to\F$: Let $\bw\in\F^k$ be drawn uniformly at random, and set $\bu$
  be the restriction of $C(\bw)$ to the set $Q$. Then, the distribution
  of $\bu$ is identical to the uniform distribution over the set of all
  functions from $Q$ to $\F$.
\end{lemma}

\subsection{Probabilistically checkable proofs of proximity ($\PCPP$)}
\label{subsec:PCPP-def}

{As described briefly in the introduction, a $\PCPP$ verifier for a property $\calP$ is given access to an input $x$ and a proof $\pi$, as well as a \emph{detection radius} $\eps>0$ and \emph{soundness error } $\delta>0$. The verifier should make a constant  number of queries (depending only on $\eps,\delta$) to the input $x$ and the proof $\pi$, and satisfy the following. If $x\in \calP$, then there exists $\pi$ for which the verifier should always accept $x$. If $\dist(x,\calP)>\eps$, the verifier should reject $x$ with probability at least $\delta$, regardless of the contents of $\pi$. More formally, we define the following.}

\begin{definition}[$\PCPP$] For $n\in\N$, let $\calP\subset \zo^n$ be a
  property of $n$-bit Boolean strings, and let $t\in \N$. We say that $\calP$ has a
  \emph{$q(\epsilon,\delta)$-query, length-$t$ Probabilistically
    Checkable Proof of Proximity ($\PCPP$) system} if the following
  holds: {There exists a verification algorithm $V$  that takes as input $\epsilon,\delta>0$ and
    $n \in \N$, makes a total of $q(\epsilon,\delta)$ queries on
    strings $w\in\zo^n$ and $\pi\in\zo^t$, and satisfies the
    following}:

	\begin{enumerate}
		\item{ (Completeness)} If $w\in \calP$, then there exists a proof $\pi={\Proof_\calP(w)}\in\zo^t$ such that for every $\epsilon,\delta>0$, the verifier $V$ accepts with probability $1$.
		\item {(Soundness) }If $\dist(w,\calP)>\eps$, then for every alleged proof $\pi\in\zo^t$, the verifier $V$ rejects with probability greater than $\delta$.
	\end{enumerate}
\end{definition}
Note that soundness is easy to amplify: Given a $\PCPP$ as above with parameters $\eps, \delta, t$ and query complexity $q(\eps, \delta)$, one can increase the soundness parameter to $1-\tau$ by simply running $\Theta(\log(1/\tau) / \delta)$ independent instances of the verification algorithm $V$, and rejecting if at least one of them rejected; the query complexity then becomes $\Theta(q(\eps, \delta) \cdot \log(1/\tau) / \delta)$, while the parameters $\eps$ and $t$ remain unchanged.

The following lemma, {establishing the existence of a quasilinear
  $\PCPP$ for any property $\calP$ that is verifiable in quasilinear
  time}, will be an important tool throughout this work.
\begin{lemma}[Corollary 8.4 in \cite{Din07}, see also
	\cite{GM07}] \label{lem:PCPP-Irit} Let $\calP$ be a property of
	Boolean strings which is verifiable by a size $t$ Boolean
	circuit. Then, there exists a length-$t'$ $\PCPP$ system
	$\calP$ with parameters $\eps,\delta > 0$, that makes at most
	$q(\epsilon,\delta)$ queries, where $t'=t\cdot \polylog{t}$.
\end{lemma}
Specifically, $q(\eps, \delta) = O(\eps^{-1})$ suffices for any $\delta < 0.99$.

As described briefly in the introduction, maximally hard properties cannot have a constant query $\PCPP$ proof systems with a sublinear length proof string.

\begin{proposition} \label{prop:NoSublinearPCPP-HardProperties}Let $\calP\subseteq \zo^n$ and $\eps>0$ be such that any $\eps$-tester for $\calP$ has to make $\Omega(n)$ many queries. Then, any constant query $\PCPP$ system for $\calP$ (where e.g.~$\delta=1/3$) must have proof length of size $\Omega(n)$.
\end{proposition}

\begin{proof} Suppose that there exists a $\PCPP$ for $\calP$ with $O(1)$ queries and proof length $t=o(n)$. Since the $\PCPP$ verifier has constant query complexity, we may assume that it is non adaptive and uses $q=O(1)$ queries. By an amplification argument as above, we can construct an amplified verifier that makes $O(q\cdot t)=o(n)$ queries, with soundness parameter $1-2^{-t}/3$. By the fact that the verifier is non-adaptive, it has the same query distribution regardless of the proof string. Therefore, we can run $2^t$ amplified
	verifiers in parallel while reusing queries, one verifier for each of the $2^t$ possible proof strings. If any of the $2^t$ amplified verifiers accept, we accept the input.
	If the input belongs to $\calP$, one of the above $2^t$ verifiers will accept
	(the one that used the correct proof). If the input was $\epsilon$-far from
	$\calP$, then by a union bound, the probability that there was any
	accepting amplified verifier is at most $1/3$. This yields an $o(n)$ tester for $\calP$, which contradicts our assumption.
\end{proof}

{
\subsection{Testing, tolerant testing and erasure-resilient testing}
In this subsection we define notions related to the property testing framework . We also formally define a few variants of the original testing model that will be addressed in this work.

A \emph{property} $\calP$ of $n$-bit boolean strings is a subset of all those
strings, and we say that a string $x$ \emph{has} the property $\calP$ if
$x\in\calP$.
%
%

Given $\eps\ge 0$ and a property $\calP$, we say that a string $x\in\zo^n$ is \emph{$\eps$-far} from $\calP$ if $\dist(x,\calP)>\eps$, and otherwise it is \emph{$\eps$-close} to $\calP$.
 We next define the notion of a \emph{tolerant} tester of which standard (i.e.~intolerant) testers are a special case.

\begin{definition} [Intolerant and tolerant testing] Given $0\le\epsilon_0<\epsilon_1\le1$, a $q$-query $(\eps_0,\eps_1)$-\emph{testing} algorithm $T$ for a property $\calP\subseteq \{0,1\}^n$ is a probabilistic algorithm (possibly adaptive) making $q$ queries to an input $x\in \zo^n$ that outputs a binary verdict satisfying the following two conditions.
	\begin{enumerate}
		\item If $\dist(x,\calP)\le \eps_0$, then $T$ accepts $x$ with probability at least $2/3$.
		\item If $\dist(x,\calP)>\eps_1$, then $T$ rejects $x$ with probability at least $2/3$.
	\end{enumerate}
When $\eps_0=0$, we say that $T$ is an \emph{$\eps_1$-testing algorithm} for $\calP$, and otherwise we say that $T$ is an \emph{$(\eps_0,\eps_1)$-tolerant testing algorithm} for $\calP$.
\end{definition}

Next, we define the erasure-resilient testing model. We start with some terminology.
A string $x\in \{0,1, \bot\}^n$ is $\alpha$-erased if $x_i$ is equal to $\bot$ on at most $\alpha n$ coordinates. A string $x'\in \zo^n$ that differs from $x$ only on coordinates $i\in[n]$ for which $x_i=\bot$ is called a \emph{completion} of $x$. The (pseudo-)distance $\dist(x, \calP)$ of an $\alpha$-erased string $x$ from a property $\calP$ is the minimum, over every completion $x'$ of $x$, of the relative Hamming distance of $x'$ from $\calP$. Note that for a string with no erasures, this is simply the Hamming distance of $x$ from $\calP$.
As before, $x$ is $\eps$-far from $\calP$ if $\dist(x, \calP) > \eps$, and $\eps$-close otherwise.

\begin{definition}[Erasure-resilient tester] Let $\alpha\in[0,1)$ and  $\eps\in(0,1)$ be parameters satisfying $\alpha+\eps<1$. A $q$-query $\alpha$-erasure-resilient $\eps$-tester $T$ for $\calP$ is a probabilistic algorithm making $q$ queries to an $\alpha$-erased string $x\in \{0,1,\perp\}^n$, that outputs a binary verdict satisfying the following two conditions.
	\begin{enumerate}
		\item If $\dist(x, \calP) = 0$ (i.e., if there exists a completion $x'$ of $x$, such that $x'\in \calP$),
		 then $T$ accepts $x$ with probability at least $2/3$.
		\item If $\dist(x, \calP) > \eps$ (i.e., if every completion of $x'$ of $x$ is $\eps$-far from $\calP$), then $T$ rejects $x$ with probability at least $2/3$.
	\end{enumerate}
\end{definition}}
\medskip

The next lemma will be useful to prove that some properties are hard
to test. The lemma states that if we have two distributions whose
restrictions to any set of queries of size at most $q$ are identical,
then no (possibly adaptive) algorithm making at most $q$ queries can
distinguish between them.

\begin{definition}[Restriction]
Given a distribution $\calD$ over functions $f: D\to \zo$ and a subset $Q\subseteq D$, we define the \emph{restriction $\calD|_Q$ of $\calD$ to $Q$} to be the distribution over functions $g:Q\to\zo $, that results from choosing a function $f:D\to \zo$ according to $\calD$, and setting $g$ to be $f|_Q$, the restriction of $f$ to $Q$.
\end{definition}


\begin{lemma} [\cite{fischer2004functions}, special
	case] \label{lem:ditsributed-exactly-indist}Let $\calD_1$ and
	$\calD_2$ be two distributions of functions over some domain $D$. Suppose
	that for any set $Q\subset D$ of size at most $q$, the
	{restricted} distributions $\calD_1|_Q$ and $\calD_2|_Q$ are identically distributed. Then, any  (possibly adaptive) algorithm making at most $q$ queries cannot distinguish $\calD_1$ from $\calD_2$ with any positive probability.
\end{lemma}
\section{Code Ensembles}
It will be necessary for us to think of a generalized definition of an encoding, in which each encoded value has multiple legal encodings.
	\begin{definition}[Code ensemble]
	A \emph{code ensemble} is a function $\calE:\Sigma^k \to 2^{\Sigma^m}$. Namely, every $x \in \Sigma^k$ has a set of its valid encodings from $\Sigma^m$.  We define the distance of the code ensemble as $$\min_{x\neq x'\in \zo^{k}}\;\min_{(v,u)\in \calE (x)\times \calE (x')}\dist(v,u).$$

	\end{definition}
	\medskip
	It is useful to think of a code ensemble $\calE : \Sigma^k\to 2^{\Sigma^m}$ as a {\em randomized mapping}, that given $x\in \Sigma^k$, outputs a uniformly random element from the set of encodings $\calE(x)$.
	Using the above we can define a \emph{shared secret} property. In particular, we use a strong information theoretic definition of a shared secret, in which $o(m)$ bits do not give \emph{any information at all} about the encoded value. Later on, we construct code ensembles with a shared secret property.

	\begin{definition}[Shared Secret] For $m,k \in \N$ and a constant $\zeta>0$, we say that a code ensemble $\calC : \zo^k\to 2^{(\zo^m)}$ has a \emph{$\zeta$-shared secret} property if it satisfies the following. For any $Q\subseteq [m]$ of size $|Q|\le \zeta  m$, any $w,w'\in\zo^{k}$ such that $w\neq w'$, and any $t\in\zo^{|Q|}$ it holds that
		\[ \Prx_{\bv \sim \calC(w)}[\bv|_Q=t]\;=\Prx_{\bv' \sim \calC(w')}[\bv'|_Q=t].\]
		Namely, for any $w\neq w'$ and any $Q\subseteq [m]$ of size at most $\zeta m$, the distribution obtained by choosing a uniformly random member of $\calC(w)$ and considering its restriction to $Q$, is identical to the distribution obtained by choosing a uniformly random member of $\calC(w')$ and considering its restriction to $Q$.\end{definition}

\subsection{A construction of a hard code ensemble}

{We describe a construction of a code ensemble for which a linear number of
  queries is necessary to verify membership or to decode the encoded
  value. This code will be our \emph{base} code in the iterative
  construction.}  The existence of such a code ensemble is proved probabilistically, 
relying on the following simple lemma.
\begin{lemma}
	\label{lem:span_random_vectors}
        Fix constant $\alpha, \beta > 0$ where
        $\beta \log (e/\beta) < \alpha$. Let $s, t \in \N$ so that
        $s \leq (1-\alpha)t$. Then, with probability $1-o(1)$, a
        sequence of $s$ uniformly random vectors
        $\{v_1, \ldots, v_s\}$ from $\{0,1\}^t$ is linearly
        independent, and corresponds to a $\beta$-distance linear code.
\end{lemma}
\begin{proof}
  The proof follows from a straightforward counting argument. If we
  draw $s$ uniformly random vectors $v_1, \ldots, v_s \in \{0,1\}^t$,
  then each non-trivial linear combination of them is in itself a
  uniformly random vector from $\{0,1\}^t$, and hence has weight less
  than $\beta$ with probability at most
$$2^{-t} \cdot \binom{t}{\beta t} \leq 2^{-t} \left( \frac{et}{\beta
    t} \right)^{\beta t} = 2^{-t} \cdot 2^{\beta \log(e/\beta) t} =
2^{(\gamma - 1)t},$$ where we set
$\gamma = \beta \log (e/\beta) < \alpha$.

By a union bound over all $2^{s} \leq 2^{(1-\alpha)t}$ possible
combinations, the probability that there exists a linear combination
with weight less than $\beta$ is at most
$2^{(\gamma-\alpha)t} = o(1)$. If this is not the case, then
$v_1, \ldots, v_s$ are linearly independent, and moreover, 
$\{v_1, \ldots, v_s\}$ corresponds to a $\beta$-distance linear code
(where we use the fact that the distance of a linear code is equal to
the minimal Hamming weight of a non-zero codeword).
\end{proof}

Our construction makes use of a sequence of vectors that correspond to a high-distance and high-dual distance code, as described below.
\begin{definition}[Hard code ensemble $\calH_k$] \label{def:HardCode}Let $k\in\N$ and let $\{v_1,\ldots,v_{3k}\}$ be a sequence of vectors in $\{0,1\}^{4k}$ such that $\mathrm{Span}\{v_1,\ldots,v_{3k}\}$ is a $1/30$-distance code, and that $\mathrm{Span}\{v_{k+1},\ldots,v_{3k}\}$ is a $1/10$-dual distance code. Let
	\[
	A=\begin{bmatrix}
	\lvert &  &\lvert \\
	v_1 & \cdots &v_{3k}\\
	\lvert & & \lvert
	\end{bmatrix}.
	\]
	We define the code ensemble $\calH_k:\zo^{k }\to 2^{\zo^{4k}}$ as
	\[\calH_k(w)=\{A u\;:\; u\in \zo^{3k}\text{ where }\;u|_{\{1,\ldots,k\}}=w \},\]
	where all operations are over $\mathrm{GF}(2)$.
\end{definition}
The next lemma states that a collection of random vectors $\{v_1, \ldots, v_{3k}\}$ in $\{0,1\}^{4k}$ satisfies the basic requirements of a code ensemble $\calH_k$ with high probability (that is, with probability tending to one as $k \to \infty$), and hence such a code ensemble exists.

\begin{lemma}
	A set $\{v_1, \ldots, v_{3k}\}$ of random vectors in $\{0,1\}^{4k}$ satisfies with high probability the following two conditions: $\mathrm{Span}\{v_1,\ldots,v_{3k}\}$ is a $1/30$-distance code, and $\mathrm{Span}\{v_{k+1},\ldots,v_{3k}\}$ is a $1/10$-dual distance code. In particular, for all $k$ large enough the code ensemble $\calH_k$ exists.
\end{lemma}
{\begin{proof}
		We apply Lemma~\ref{lem:span_random_vectors} multiple times. First, picking $t = 4k$, $s = 3k$, $\alpha=1/4$, and $\beta=1/30$, we conclude that $v_1, \ldots, v_{3k}$ with high probability correspond to a $1/30$-distance code.
%

		To show that with high probability the code
                spanned by the last $2k$ vectors has high dual
                distance, we compare the following two processes,
                whose output is a linear subspace of
                $(\mathrm{GF}(2))^{4k}$, that we view as a code: (i)
                Choose $2k$ vectors and return their span. (ii) Choose
                $4k-2k=2k$ vectors and return the dual of their span.
                Conditioning on the chosen $2k$ vectors being linearly
                independent, the output distributions of these two
                processes are identical. Indeed, by a symmetry argument it is not hard to see that under the conditioning, the linear subspace generated by Process (i) is uniformly distributed among all rank-$2k$ subspaces $V$ of $(\mathrm{GF}(2))^{4k}$. Now, since we can uniquely couple each such $V$ with its dual $V^\perp$ (also a rank-$2k$ subspace) and since $V = (V^{\perp})^{\perp}$, this means that the output distribution of Process (ii) is uniform as well.

		However, it follows again from
		Lemma~\ref{lem:span_random_vectors} (with $t = 4k$,
		$s = 2k$, $\alpha = 1/2$, and any $\beta > 0$ satisfying the conditions of the lemma) that the chosen $2k$ vectors are independent
		with high probability.
		This means that (without the conditioning) the output
                distributions of Process (i) and Process (ii) are
                $o(1)$-close in variation distance. Applying Lemma~\ref{lem:span_random_vectors} with $t = 4k$, $s = 2k$,
                $\alpha = 1/2$, and $\beta = 1/10$ we get that the distance of the code generated by Process (i) is at least
                $\beta = 1/10$ with high probability. However, the latter distance equals by definition to the dual distance of the code generated by Process (ii). By the closeness of the distributions, we conclude that the dual distance of Process~(i) is also at least $1/10$ with high probability.
\end{proof}}

\medskip

%

{
	We next state a simple but important observation regarding membership verification.
	\begin{observation} \label{obs:HardCodeProperties} Once a matrix $A$
		with the desired properties is constructed (which may take
		$\exp(k^2)$ time if we use brute force), given $w\in\zo^k$, the membership of $v$ in $\calH_k(w)$ can be verified in $\poly(k)$ time (by solving a system of linear equations over $\mathrm{GF}(2)$).

\end{observation}}

\section{ $\PCU$s and $\PCUSS$s}\label{sec:PCUandPCUSS}

\medskip {Next, we define the notion of Probabilistically Checkable
  Unveiling ($\PCU$). This notion is similar to $\PCPP$, but here instead of requiring our input to satisfy a given property, we require our input to encode a value $w \in \{0,1\}^k$ (typically
  using a large distance code ensemble). We
  then require that given the encoded value $w$, it will be possible
  to prove in a $\PCPP$-like fashion that the input is indeed a valid
  encoding of $w$.

\begin{definition}[$\PCU$] \label{def:PSU} Fix $m,t,k\in \N$, and let $\calC: \{0,1\}^k \to 2^{\{0,1\}^m}$ be a code ensemble. We say that $\calC$ has a $q(\epsilon,\delta)$-query, length-$t$ $\PCU$ if the following holds. There exists a verification algorithm $V$ that takes as inputs $\epsilon,\delta>0$, $m\in \N$, and $w\in\{0,1\}^k$, makes at most $q(\epsilon,\delta)$ queries to the strings $v\in\zo^m$ and $\pi\in\zo^t$, and satisfies the following:
	\begin{enumerate}
		\item If $v\in\calC(w)$, then there exists a proof $\pi=\Proof_{\calC}(v)\in\{0,1\}^t$ such that for every $\epsilon,\delta>0$, the verifier $V$ accepts with probability $1$.
		\item If $\dist(v,\calC(w))>\eps$, then for every alleged proof $\pi\in\zo^t$, the verifier $V$ rejects $v$ with probability greater than $\delta$.
	\end{enumerate}
\end{definition}

In order to facilitate the proof of the main theorem, we utilize a more stringent variant of the above $\PCU$ definition. Instead of supplying $w\in\zo^k$ to the algorithm, we supply  oracle access to a a string $\tau\in\zo^{100k}$ that is supposed to represent $\Spiel(w)$, along with the proof $\pi$, and the algorithm only makes $q(\epsilon,\delta)$ queries to
the proof string $\pi$, the original encoding $v$ \emph{and} the string $\tau$.
For cases where $v\in \calC (w)$, we use $\Value(v)$ to denote $\Spiel(w)$.
\begin{definition}[$\Spiel$-$\PCU$] \label{def:SpielPSU} Fix $m,t,k\in \N$, and let $\calC: \{0,1\}^k \to 2^{\{0,1\}^m}$ be a code ensemble. We say that $\calC$ has a $q(\epsilon,\delta)$-query, length-$t$ $\Spiel$-$\PCU$ if the following holds. There exists a verification algorithm $V$  that takes as inputs $\epsilon,\delta>0$, $m\in \N$, makes at most $q(\epsilon,\delta)$ queries to the strings $v\in\zo^m$, $\tau\in\zo^{100k}$ and $\pi\in\zo^t$, and satisfies the following:
	\begin{enumerate}
		\item If there exists $w\in\zo^k$ for which $v\in\calC(w)$ and $\tau=\Value(v)=\Spiel(w)$, then there exists a proof $\pi=\Proof_{\calC}(v)\in\{0,1\}^t$ such that for every $\epsilon,\delta>0$, the verifier $V$ accepts with probability $1$.
		\item If for every $w\in\zo^k$ either $\dist(\tau,\Spiel(w))>\eps$ or  $\dist(v,\calC(w))>\eps$, then for every alleged proof $\pi\in\zo^t$, the verifier $V$ rejects $v$ with probability greater than $\delta$.
	\end{enumerate}
\end{definition}

 Note that a code ensemble admitting a $\Spiel$-$\PCU$ automatically admits a $\PCU$. Indeed, given the string $w$, an oracle for
$\Spiel(w)$ can be simulated.

\medskip
{The following lemma states the existence of $\Spiel$-$\PCU$ for efficiently computable code ensembles, and will be used throughout this work. The proof follows from Lemma~\ref{lem:PCPP-Irit} together with a simple concatenation argument.

 \begin{lemma}\label{lem:Quasilinear->PCU} Let $k,m,t\in \N$ be such that $t\ge m$, and let $\calC: \zo^k \to 2^{\zo^m}$ be a code ensemble. If given  $w\in\zo^k$ and $v\in\zo^m$, it is possible to verify membership of $v$ in $\calC (w)$ using a circuit of size $t$, then there is a $q(\eps,\delta)$-query, length-$t'$ $\Spiel$-$\PCU$ for $\calC$ where $t'=t\cdot \polylog t$.
 \end{lemma}

\begin{proof} Assume without loss of generality that $m\ge |\Spiel(0^k)|$. Let $\xi =\left\lfloor \frac{m	}{|\Spiel(0^k)|}\right\rfloor$ (note that $\xi\ge 1$), and define
	\[ \calC_{eq}\eqdef \left\{v\sqcup (\Spiel(w))^{\xi}\;\big\lvert \;\exists w\in \zo^k\text{ for which}\; v\in \calC (w) \right\},  \]
	where $(\Spiel(w))^\xi$ denotes the $\xi$-times concatenation of $\Spiel(w)$.

	For any string $u$ it is possible to check, using a quasilinear size circuit (see \cite{Spi96}), that the substring that corresponds to the domain of $(\Spiel(w))^\xi$ is a $\xi$-times repetition of $\Spiel(w)$ for some $w$. After doing so, we decode $w$ using a quasilinear size circuit (as in \cite{Spi96}), and then, by the premise of the lemma, we can verify membership in $\calC (w)$ using a circuit of size $t$. Therefore, membership in $\calC_{eq}$ can be decided using a $O(t)$ size boolean circuit, and therefore by Lemma~\ref{lem:PCPP-Irit} admits a $\PCPP$ system whose proof length is  quasilinear in $t$.

Given an input $v$ to $\Spiel$-$\PCU$, let $v'=v\sqcup(\Spiel(w))^\xi$ and use the $\PCPP$ system for $\calC_{eq}$, with detection radius $\eps/3$ and soundness parameter $\delta$, where each query to $v'$ is emulated by a corresponding query to $v$ or $\Spiel(w)$.
Note that if $v\in\calC (w)$, then $v'\in\calC_{eq}$, so the $\PCPP$ system for $\calC_{eq}$ will accept with probability $1$.

Next, suppose that $\dist(v,\calC (w))>\eps$, and observe that this implies that $v'$ is at least $\eps/3$-far from $\calC_{eq}$. Thus, by the soundness property of the $\PCPP$ for $\calC_{eq}$, the verifier rejects with probability at least $\delta$, regardless of the contents of the alleged proof $\pi$ it is supplied with.
\end{proof}
\medskip}

Next we define Probabilistically Checkable Unveiling of a Shared Secret ($\PCUSS$).
\begin{definition} \label{def:PCUSS}For $m,k,t\in \N$, we say that a function $\calC\colon\zo^k \to 2^{(\zo^n)}$ has a $q(\epsilon,\delta)$-query, length-$t$ $\PCUSS$, if $\calC$ has a shared secret property, as well as $\calC$ has a $q(\epsilon,\delta)$-query, length-$t$ $\PCU$.
Similarly, when $\calC$ has a shared secret property (for constant $\zeta$), as well as $\calC$ has a $q(\eps,\delta)$-query, length-$t$ $\Spiel$-$\PCU$, we say that $\calC$ has a $q(\eps,\delta)$-query, length-$t$ \emph{$\Spiel$-$\PCUSS$}.
\end{definition}}
\noindent Note that $\calC$ admitting a $\Spiel$-$\PCUSS$ directly implies that it admits a $\PCUSS$ with similar parameters.

\medskip
The following lemma establishes the existence of a $\Spiel$-$\PCUSS$ for $\calH_k$, where $\calH_k$ is the code ensemble from Definition~\ref{def:HardCode}.
\begin{lemma}\label{lem:HardCodePCUSS} For any $k\in\N$, $\calH_k$ has a $q(\eps,\delta)$-query, length-$t'$ $\Spiel$-$\PCUSS$ where $t'=\poly(k)$.
\end{lemma}
\begin{proof} By Observation~\ref{obs:HardCodeProperties}, given $w$, membership in $\calH_k(w)$ can be checked in $\poly(k)$ time, which means that there exists a polynomial size circuit that decides membership in $\calH_k(w)$. Combining the above with Lemma~\ref{lem:Quasilinear->PCU} implies a $q(\eps,\delta)$-query, length-$t'$ $\Spiel$-$\PCU$ where $t'=\poly(k)$. By Lemma~\ref{lem:dual-to-secret}, the large dual distance property of $\calH_k$ implies its shared secret property for some constant $\zeta$, which concludes the proof of the lemma.
\end{proof}


\section{$\PCUSS$ construction}

In this section we give a construction of code ensembles that admit a $\PCUSS$. First we show that our code ensemble has a $\PCU$ with a short proof. Specifically,

\begin{lemma} \label{lem: PCU-construction}For any  fixed $\ell\in\N$ and any $k\in \N$, there exists $n_0(\ell,k)$ and a code ensemble $\calE^{(\ell)}:\zo^{k}\to 2^{(\zo^n)}$, such {that for all $n>n_0(\ell,k)$,} the code ensemble $\calE^{(\ell)}$ has a $q(\eps,\delta)$-query length-$t$ $\PCU$, for $t=O(n\cdot\polylogell n)$.
\end{lemma}

Later, we prove that our code ensemble has a shared secret property, which implies that it has a $\PCUSS$ (which implies Theorem~\ref{informalthm:main}, as we shall show).

\begin{theorem}\label{thm:PCUSS-construction} For any fixed $\ell\in\N$ and any $k\in \N$, there exists $n_0(\ell,k)$ and a code ensemble $\calE^{(\ell)}:\zo^{k}\to 2^{(\zo^n)}$, {such that for all $n>n_0(\ell,k)$}, the code ensemble $\calE^{(\ell)}$ has a $q(\eps,\delta)$-query length-$t$ $\PCUSS$, for $t=O(n\cdot\polylogell n)$.
\end{theorem}
\noindent Specifically, by the discussion before Lemma \ref{lem:Iter-completeness}, for any fixed soundness parameter $0 < \delta < 1$ it suffices to take
$$q(\eps, \delta) \leq \left(2^\ell/\eps\right)^{O(\ell)},$$
and for the high soundness regime where $\delta = 1-\tau$ (and $\tau > 0$ is small), it suffices to have
$$q(\eps, \delta) \leq \left(2^\ell/\eps\right)^{O(\ell)} \log(1/\tau).$$

\subsection{The iterated construction}

Our iterative construction uses polynomials over a binary finite field
$\mathrm{GF}(2^t)$. In our proof we will need to be able to implement
arithmetic operations over this field efficiently (i.e., in $\poly(t)$
time). This can be easily done given a suitable representation of the
field: namely, a degree $t$ irreducible polynomial over
$\mathrm{GF}(2)$. It is unclear in general whether such a polynomial
can be found in $\poly(t)$ time. Fortunately though, for
$t=2\cdot 3^r$ where $r\in\N$, it is known that the polynomial
$x^t+x^{t/2}+1$ is irreducible over $\mathrm{GF}(2)$ (see,
e.g.,~\cite[Appendix G]{Gol08}). We will therefore restrict our
attention to fields of this form. At first glance this seems to give
us a property that is defined only on a sparse set of input
lengths. However, towards the end of this section, we briefly describe
how to bypass this restriction.


We next formally define our iterated construction, starting with the ``level-0'' construction as a base case. The constants $c,d$ in the definition will be explicitly given in the proof of Lemma~\ref{lem:Iter-length}. Additionally, for any $\ell \in \N$, we shall pick a large enough constant $c_{\ell}$ that satisfies several requirements for the ``level-$\ell$'' iteration of the construction.

%

\begin{definition}[Iterated coding ensemble]\label{def:iteratedEnc}
  For $k\in\N$ and $w\in\zo^{k}$, we define the \emph{base code ensemble of
    $w$} (i.e., level-$\ell$ code ensemble of $w$ for $\ell=0$)
  as \[\calE_k^{(0)}{(w)}=\calH_k(w).\]
  Let $c,d\in \N$ be large
  enough global constants,  fix $\ell>0$, let $c_{\ell}$ be large enough, and let $\F$ be a
  finite field for which $|\F|\ge \max\{c_\ell, c\cdot k\}$.
\medskip

We define the \emph {level-$\ell$ code ensemble of $w\in\zo^{k}$ over
    $\F$} as follows. {Let $r\in \N$ be the smallest integer such that $(\log|\F |)^d\le 2^{2\cdot 3^r}$, set $\F'=\mathrm{GF}\left(2^{2\cdot 3^r}\right)$ and
  $k'=\log|\F|$. Note that these satisfy the recursive requirements of
  a level-$(\ell-1)$ code ensemble provided that $c_\ell$ is large enough
  (specifically we require
  $(\log|\F|)^{d-1}>c$, so that $|\F'|\ge ck'$)}. Finally, let $H\subseteq \F$
  be such that $|H|=k$, and define
  \[\calE^{(\ell)}_{\F,k}{(w)}=\bigcup_{g\in
      \calC_{\F}:\;g|_{H}=w\;}\bigsqcup_{\beta\in \F\setminus
      H}\calE_{\F ',k'} ^{(\ell-1)}{(\qu{g(\beta)})}.
 \](Note that for $\ell=1$ we just use $\calE^{(1)}_{\F,k}{(w)}=\bigcup_{g\in
 	\calC_{\F}:\;g|_{H}=w\;}\bigsqcup_{\beta\in \F\setminus
 	H}\calE_{k'} ^{(0)}{(\qu{g(\beta)})}$).

\end{definition}

{That is, $v\in \calE_{\F,k}^{(\ell)}(w)$ if there exists a polynomial
$g\in \calC_{\F}$ such that
  $v= \bigsqcup_{\beta\in\F \setminus H} v_\beta$, where $v_\beta \in \calE_{\F ' ,k'}^{(\ell-1)}(\qu {g(\beta)})$ for every $\beta\in\F \setminus H$ and $g|_H=w$ (where we identify the $0$ and $1$ elements of $\F$ with $0$ and $1$ bits respectively). When the context is clear, we sometimes omit the subscripts.}
\medskip

Our choice of the constants $c,d,c_\ell$ needs to satisfy the following conditions. The constant $c$ is chosen such that $H$ will not be an overly large portion of $\F$ (this requirement is used in Lemma~\ref{lem:BaseDistance}). The constant $d$ is needed to subsume the length of $\PCPP$ proof string which is part of the construction (this requirement is used in Lemma~\ref{lem:Iter-length}). Finally, the constant $c_\ell$ needs to be large enough to enable iteration (as explained in Definition~\ref{def:iteratedEnc} itself).

  \medskip
Let $\ell\ge 0$ be some fixed iteration. The following simple
observation follows by a simple inductive argument using the
definition of the level-$\ell$ coding ensemble, and in particular that
$|\F'|=\polylog|\F|$.

\begin{observation} For $\ell>0$, let $n=|\F|$ and $w\in\zo^{k}$. If $v\in \calE^{(\ell)}{(w)}$, then $m^{(\ell)}_\F \eqdef|v|=n\cdot \poly(\log n)\cdot \poly(\log\log n)\cdots\poly(\log^{(\ell)}n)$, where $\log^{(\ell)}n$ is the $\log$ function iterated $\ell$ times.
\end{observation}

%
%
%
%
%

When the field $\F$ is clear from context, we shall usually write $m^{(\ell)}$ as a shorthand for $m^{(\ell)}_\F$.
The following lemma, proved in the next subsection, establishes the existence of short length $\Spiel$-$\PCU$s for our code ensembles.
\begin{lemma}\label{lem:IterPCPP}For any $\ell\ge 0$, the code ensemble
	$\calE_{\F ,k}^{(\ell)}$ admits a $q(\eps,\delta)$-query, length-$t$ $\Spiel$-$\PCU$ for $t=O(m^{(\ell)}\cdot \polylogell m^{(\ell)})$.
\end{lemma}

\subsection{Proof of Lemma~\ref{lem:IterPCPP}}
\label{subsec:prof_of_main_lemma}
{We start by defining the $\PCU$ proof string for a given  $v\in \calE_{\F ,k}^{(\ell)}(w)$ for some $w\in\zo^k$.}

\begin{definition}[The $\PCU$ Proof String] \label{def:ProofP_l} For
  $\ell=0$, let $v\in\calE_{k}^{(0)}(w)$ and $\Value^{(0)}(v)=\Spiel(w)$. We define the proof string for $v$, $\Proof^{(0)}(v)$, as the one guaranteed by Lemma~\ref{lem:HardCodePCUSS}  (note that the length of  $\Proof^{(0)}(v)$ is $\poly(k)$).
\medskip



  For $\ell>0$, let $g\in \calC_{\F}$ and $w\in\zo^k$ be such that
  $v\in \bigsqcup_{\beta\in \F\setminus H} \calE ^{(\ell-1)}(\qu
  {g(\beta)})$, $\Value^{(\ell)}(v)=\Spiel(w)$ and $g|_H=w$. In addition, set
  $S_v\eqdef\bigsqcup_{\beta\in
    \F\setminus H}\Value^{(\ell-1)}(v_\beta)=\bigsqcup_{\beta\in \F\setminus H}\Spiel( {g(\beta)})$. The proof string for
  $v\in\calE_{\F ,k}^{(\ell)}$ is defined as follows.
  \[
    \Proof^{(\ell)}(v)= S_v \sqcup\bigsqcup_{\beta\in\F \setminus H} \Proof ^{(\ell-1)}(v_\beta)\sqcup \Proof_{\calL}\left(S_v\right)\]
 where the code ensemble
  $\calL:\zo^k \to 2^{\zo^{O(|\F|\cdot \log|\F |)}}$ is defined as follows. Given $w\in \zo^k$, $S\in \calL(w)$ if and only
  if there exists a polynomial $g\in \calC_{\F}$ such that the
  following conditions are satisfied.
  \begin{enumerate}
    \item $g|_H=w$.
    \item $S=\bigsqcup_{\beta\in \F\setminus H} \Spiel(g(\beta))$.
\end{enumerate}

\end{definition}

The following lemma establishes the existence of a $\Spiel$-$\PCU$ for $\calL$.

\begin{lemma} \label{lem:Spiel-PCU-for-L} $\calL$ has a $q(\eps,\delta)$-query length-$t$ $\Spiel$-$\PCU$ for $t=O(|\F |\cdot\polylog |\F |)$.
\end{lemma}
\begin{proof} By Theorem~\ref{thm:spiel}, there exists a quasilinear size circuit that decodes $\Spiel( {\alpha})$. Using such a circuit, we can decode $g(\beta)$ from $S$ for every $\beta\in \F$. Then, using all the values $g(\beta)$ and $w$ (where the $i$-th bit of $w$ correspond to the value of the $i$-th element in $H$ according to the ordering), we use Theorem~\ref{lem:FastInterpolation} to interpolate the values and achieve a representation of a polynomial $g:\F\to\F$. If $g\in \calC_{\F}$ we accept $S$ and otherwise we reject. Since deciding if $S\in\calL(w)$ has a quasilinear size circuit, by Lemma~\ref{lem:Quasilinear->PCU}, there is a quasilinear length $\Spiel$-$\PCU$ for $\calL$.
\end{proof}

\medskip


Having defined $\Proof^{(\ell)}$, we first provide an upper bound on the bit
length of the {prescribed} proof string.  For $\ell> 0$, let $z^{(\ell)}_{\F,k}$ 
denote the bit length of the proof for membership in $\calE^{(\ell)}$ as defined in
Definition~\ref{def:ProofP_l}, where for $\ell=0$ we replace the
(nonexistent) field $\F$ with $|w|$.

\medskip

The following lemma, establishing the proof string's length, relies on our choice of the constant $d$ in Definition~\ref{def:iteratedEnc}. In particular, $d$ needs to be large enough to subsume the size of $\Proof_\calL(\cdot)$

\begin{lemma} \label{lem:Iter-length}For any $\ell\ge 0$, we have that
  $z^{(\ell)}_{\F,k}=O \big(m^{(\ell)}\cdot\polylogell m^{(\ell)} \big)$.
\end{lemma}
\begin{proof} The proof follows by induction on $\ell$. The base case ($\ell=0$) follows directly from the definition of $\calP^{(0)}$ by our convention that  $\log^{(0)}|w|=|w|$.

 Consider $\ell>0$, and note that since the size of $S_v$ is $O(|\F|\log|\F |)$, the size of $\Proof_{\calL}(S_v)$ is  $O(|\F | \cdot \polylog |\F |)$. By combining the above with the definition of the proof string we have \[z^{(\ell)}_{\F,k}\le |\F|\cdot \polylog|\F|+|\F|\cdot z^{(\ell-1)}_{\F',k'}.\]
Now, assume that $z^{(\ell-1)}_{\F',k'}=O(m^{(\ell-1)}\cdot\mathrm{polylog}^{(\ell-1)}|\F'|)$.
Note that since the global constant $d$ was chosen so that $|\F |\cdot |\F'| \ge |\Proof_{\calL}(S_v)|$, we have that $|\F |\cdot z^{(\ell-1)}_{\F',k'}\ge |\F |\cdot |\F'| \ge  |\Proof_{\calL}(S_v)| $. Therefore,
 \[m^{(\ell)}=\Theta(|\F |\cdot m^{(\ell-1)})=\Omega(|\F|\cdot|\F'| )=\Omega(|\F|\cdot\polylog|\F|),\]
 so that $|\F|\cdot\polylog|\F|=O(m^{(\ell)})$, and

$$z^{(\ell)}_{\F,k}=O(|\F |\cdot z^{(\ell-1)}_{\F'}).$$
In addition, by the fact that $m_{\ell}=\Theta(|\F|\cdot m^{(\ell-1)})$ and the induction hypothesis we obtain
\[|\F|\cdot z^{(\ell-1)}_{\F',k'}=O(|\F|\cdot m^{(\ell-1)}\cdot \mathrm{polylog}^{(\ell-1)}|\F'| )=O(m^{(\ell)}\cdot \polylogell |\F|)=O(m^{(\ell)}\cdot \polylogell m_{\ell}).\]
So overall, we get that $z^{(\ell)}_{\F,k}=O(m^{(\ell)}\cdot\polylogell m^{(\ell)})$ as required.
\end{proof}

\medskip
Next, for an alleged proof $\pi=\Proof^{(\ell)}(v)$, we use the notation  $\pi|_{\Dom(X)}$ to denote the restriction of $\pi$ to the bits that correspond to $X$ in $\pi$ as defined in Definition~\ref{def:ProofP_l}. For example, $\pi|_{\Dom(\Value^{(\ell-1)}(v_\beta))}$ refers to the bits that represent $\Value^{(\ell-1)}(v_\beta)$.

We introduce the verifier procedure for $\calE_{\F  ,k}^{(\ell)}$ (see Figure~\ref{fig:VerifierProcedureIterated}), and prove its completeness and soundness. For technical considerations, the verifier procedure is only defined when the soundness parameter $\delta$ is small enough (as a function of $\ell$); the soundness amplification argument from Subsection \ref{subsec:PCPP-def} can easily take care of the situation where $\delta$ is larger, by running sufficiently many independent instances of the verification step.
\begin{figure}[ht!]
	\begin{framed}
		\noindent \underline{$\texttt{Verifier-Procedure}_{\;\calE^{(\ell)}}$}
		\begin{flushleft}
			\noindent {\bf Input:}  Parameters
			$\eps,\delta\in (0,1)$, an input
			$v\in\{0,1\}^{m^{(\ell)}}$, an alleged value $\tau\in \zo^{100k}$ of $v$, and an alleged proof
			$\pi\in\{0,1\}^{z^{(\ell)}_{\F,k}}$
			for $v$.

			\begin{enumerate}
				\item If $\ell=0$, use the $\PCU$ for $\calE ^{(0)}$ with parameters $\eps$ and $\delta$.
				\item If $\ell>0$:
			\begin{enumerate}
				\item Use the $\PCU$ verifier for $\calL$
				with radius $\epsilon/300$ and soundness
				$\delta$, to verify the unveiling of $\pi|_{\Dom(S_v)}$, using $\tau$ as the value oracle and $\pi|_{\Dom(\Proof_\calL(S_v))}$ as the proof oracle.\label{step1iter}
				\item For ${6}/{\epsilon}$ many times:\label{step 2iter}
				\begin{enumerate}
					\item Pick $\bbeta\in \F\setminus H$ uniformly at random.
					\item Use the $\PCU$ verifier procedure for $\calE^{(\ell-1)}$ with parameters $\eps/3$ and $2\delta$, to verify the unveiling of  $v_{\bbeta}$, using $\pi|_{\Dom(\Value^{(\ell-1)}(v_\beta))}$ as the value oracle and
					$\pi|_{\Dom(\Proof^{(\ell-1)}(v_{\bbeta}))}$ as the proof oracle.
				\end{enumerate}
			\end{enumerate}
			\end{enumerate}
			If any of the stages rejected then \textbf{Reject}, and otherwise \textbf{Accept}.

		\end{flushleft}\vskip -0.14in
	\end{framed}\vspace{-0.25cm}
	\caption{Description of $\texttt{Verifier-Procedure}_{\;\calE^{(\ell)}}$.\vspace{-0.25cm}} \label{fig:VerifierProcedureIterated}
\end{figure}

\noindent Before proceeding to the completeness and soundness proofs, let us analyze the query complexity. Denote by $Q_{\ell}(\eps, \delta)$ the query complexity of the verifier in the above procedure for a given $\ell \geq 0$. 
It follows from the recursive description of $\texttt{Verifier-Procedure}_{\;\calE^{(\ell)}}$ and the proof of Lemmas \ref{lem:Quasilinear->PCU} and \ref{lem:Spiel-PCU-for-L} that the query complexity satisfies the recurrence relation $Q_{\ell}(\eps, \delta) \leq O(1/\eps) \cdot Q_{\ell-1}(\eps/O(1)), \delta \cdot O(1)) + q^*(\Theta(\eps), \Theta(\delta))$, where $q^*(\eps^*, \delta^*) = O((\eps^*)^{-1})$ is the query complexity of Dinur's $\PCP$ \cite{Din07} with detection radius $\eps^*$ and soundness parameter $\delta^* \leq 1/2$; and furthermore, that $Q_{0}(\eps, \delta) \leq q^*(\Theta(\eps), \Theta(\delta))$.
Thus, we conclude by induction that, provided that $\delta \leq 2^{-\ell-1}$,
$$Q_{\ell}(\eps, \delta) \leq \frac{C}{\eps} \cdot \frac{C^2}{\eps} \cdot \ldots \cdot \frac{C^\ell}{\eps} \cdot q^*(\eps / 2^{O(\ell)}, \delta \cdot 2^{\ell}) = 2^{O(\ell^2)} \eps^{-O(\ell)},$$
where $C > 0$ is a large enough absolute constant. To achieve any given soundness $\delta > 1/2$, we can amplify by repeating the verifier procedure with parameter $\delta' = 2^{-\ell-1}$ a total of $2^{O(\ell)} \cdot \log\left((1-\delta)^{-1}\right)$ times and rejecting if any of these instances rejected. The query complexity is bounded by
$$2^{O(\ell^2)} \eps^{-O(\ell)} \cdot 2^{O(\ell)} \cdot \log\left((1-\delta)^{-1}\right) = (2^\ell / \eps)^{O(\ell)} \cdot \log\left((1-\delta)^{-1}\right),$$
as desired. The next two lemmas establish the completeness and soundeness of the verifier procedure, respectively.
\begin{lemma}\label{lem:Iter-completeness} If there exist $w\in \zo^k$ for which $v\in \calE_{\F ,k}^{(\ell)}(w)$, then $\emph{\texttt{Verifier-Procedure}}_{\;\calE^{(\ell)}}$ accepts $v$ with probability $1$ when supplied with oracle access to the corresponding $\Proof^{(\ell)}(v)$  and $\tau=\Value^{(\ell)}(v)=\Spiel(w)$.
\end{lemma}

\begin{proof} The proof follows by induction on $\ell$. {The base case follows directly from Lemma~\ref{lem:HardCodePCUSS}. Hence, the verifier for $\calE^{(0)}$ supplied with $\Proof^{(0)}(v)$ as the proof oracle and $\Value^{(\ell)}(v)$ as the value oracle, will accept $v$ with probability $1$.}

	Assume that $\texttt{Verifier-Procedure}_{\;\calE^{(\ell-1)}}$ accepts with probability $1$ any valid encoding $v'$ when supplied with the corresponding oracles for $\Value^{(\ell-1)}(v')$ and $\Proof^{(\ell-1)}(v')$. Let $v\in \calE^{(\ell)}$ and write $ v= \bigsqcup_{\beta\in \F\setminus H} v_\beta$, where there exist $w\in\zo^{k}$ and $g\in \calC_{\F}$ such that for all $\beta\in\F\setminus H$, $v_\beta\in \calE^{(\ell-1)}({g(\beta)})$, where $g|_H=w$ and $\tau=\Value^{(\ell)}(v)=\Spiel(w)$. Then, by the definition of the language $\calL$ and the first two components of $\Proof^{(\ell)}(v)$,  Step~(\ref{step1iter}) of $\texttt{Verifier-Procedure}_{\;\calE^{(\ell)}}$ will always accept. In addition, for every $\beta\in\F \setminus H$, we have that $v_\beta\in \calE^{(\ell-1)}$, and therefore by the induction hypothesis, Step~(\ref{step 2iter}) of $\texttt{Verifier-Procedure}_{\;\calE^{(\ell)}}$ will accept the corresponding unveiling for any picked $\beta\in\F \setminus H$.
\end{proof} 

\begin{lemma} \label{lem:Iter-soundness}If for every $w\in \zo^k$ either $\dist(\tau,\Spiel(w))>\eps$ or  $\dist(v,\calE^{(\ell)}(w))>\eps$ (or both), then with probability greater than $\delta$, $\emph{\texttt{Verifier-Procedure}}_{\;\calE^{(\ell)}}$ will reject $v$ regardless of the contents of the supplied proof string.
\end{lemma}

\begin{proof} Let $\tau\in \zo^{100k}$ be an alleged value for $v$, and $\pi\in \{0,1\}^{z^{(\ell)}_{\F,k}}$ be an alleged proof string for $v$. We proceed by induction on $\ell$. {For $\ell=0$ we use the $\PCU$ verifier for $\calE^{(0)}$ with error $\eps$ and soundness $\delta$ to check that $v$ is a member of the code ensemble $\calE^{(0)}$ and $\tau$ is its value. If the $\PCU$ verifier for $\calE^{(0)}$ rejects with probability at most $\delta$, then there exist $w\in\{0,1\}^k$ such that $\dist(v,\calE ^{(0)}(w))\le \eps$ and $\dist(\tau,\Spiel(w))\le \eps$, and the base case is complete.}

Next assume that the lemma holds for $\ell-1$. If the $\PCU$  verifier for $\calL$ in Step~(\ref{step1iter}) rejects with probability at most $\delta$, then there exist a function $g\in\calC_{\F}$ and $w\in\zo^{k}$ for which {$g|_H=w$} so that
\[\dist(\pi|_{\Dom(S_v)},\Spiel({g|_{\F \setminus H}}))\le\epsilon/300 \qquad\text{and}\qquad\dist(\tau,\Spiel(w))\le \epsilon/300.\]
In particular, the leftmost inequality means that for at most $\frac{\eps}{3}|\F \setminus H|$ of the elements  $\beta\in\F\setminus H$, it holds that
\[\dist(\pi|_{\Dom(\Value^{(\ell-1)}(v_\beta))}, \Spiel({g(\beta)})>1/100.\]
We refer to elements $\beta\in \F \setminus H$ satisfying the above inequality as \emph{bad} elements, and to the rest as \emph{good} elements. Let $G$ denote the set of good elements.

Next, we show that if the loop that uses the $\PCU$  verifier for $\calE^{(\ell-1)}$ in Step~(\ref{step 2iter}) rejects with probability at most $\delta$, then for at most an $\epsilon/3$ fraction of the good $\beta\in\F \setminus H$, it holds that \[ \dist\left(v_\beta,\calE^{(\ell-1)}({\qu {g(\beta)}})\right) >\epsilon/3.\]
Assume that there are more than $\frac{\eps}{3}\cdot |G|$ good elements such that $\dist\left(v_\beta,\calE^{(\ell-1)}({\qu {g(\beta)}})\right) >\epsilon/3$. Then, by our induction hypothesis, each of them will be rejected by the $\PCU$ verifier for $\calE^{(\ell-1)}$ with probability more than $2\delta$. In addition, with probability at least $1/2$ we sample at least one such good $\beta$, and then during this iteration the verifier in Step~(\ref{step 2iter}(ii)) rejects with conditional probability more than $2\delta$, and hence the verifier will reject with overall probability more than $\delta$.
Summing everything up, when the input is rejected with probability at most $\delta$,
\[\dist\left(v,\bigsqcup_{\beta\in \F\setminus H}\calE^{(\ell-1)}(\qu{g(\beta)})\right)\le\epsilon/3+(1-\epsilon/3)\cdot \epsilon/3+(1-\eps/3)^2\cdot \epsilon/3\le \eps, \]
where the three summands are respectively the contribution to the distance of the bad elements, the good elements with $v_\beta$ being far from any level $\ell-1$ encoding of $\qu{g(\beta)}$, and all the other elements.
\end{proof}

\medskip

The proof of {Lemma~\ref{lem:IterPCPP}} follows directly by combining Lemma~\ref{lem:Iter-length}, Lemma~\ref{lem:Iter-completeness} and Lemma~\ref{lem:Iter-soundness}.

\medskip
{\noindent The following corollary follows directly from Lemma~\ref{lem:IterPCPP} and the  definition of $\Spiel$-$\PCU$ (Definition~\ref{def:SpielPSU}), and implies Lemma~\ref{lem: PCU-construction}.
\begin{corollary}\label{cor:Pl-PCU}Let $\F$ be a finite field and $k\in \N$ which satisfy the requirements in Definition~\ref{def:iteratedEnc}. Then, for every $\ell\ge 0$ the coding ensemble $\calE_{\F,k}^{(\ell)}:\zo^k\to 2^{\left(\zo^{m^{(\ell)}}\right)}$ has a  $q(\epsilon,\delta)$-query, length-$t$ $\Spiel$-$\PCU$ for $t=O(m^{(\ell)}\polylogell m^{(\ell)})$.
\end{corollary}}
\subsection{The Lower Bound}
We turn to prove the linear query lower bound for the testability of our property.
%
We start by defining distributions over strings of length $m^{(\ell)}$.
\paragraph{Distribution $\Dyes^{(\ell)}(w)$:} Given $w\in\zo^k$, we define the distribution $\Dyes^{(\ell)}(w)$ to be the uniform distribution over elements in $\calE ^{(\ell)}(w)$.

\paragraph{Distribution $\Dno^{(\ell)}$:} An element $v$ from $\Dno^{(\ell)}$ is drawn by the following process. For $\ell=0$, $\bv$ is a uniformly random string in $\zo^{4k}$. For $\ell>0$, we pick a uniformly random function $\blambda : \F \setminus H\to \F$, and let $v$ be a uniformly random element of $\bigsqcup_{\beta\in \F\setminus H}\calE ^{(\ell-1)}(\qu {\blambda (\beta)})$

\begin{lemma}\label{lem:StrongerLB} For any $\ell\ge 0$, every $w\in\zo^k$ and $q=o(m^{(\ell)}/10^\ell)$, any algorithm making at most $q$ queries cannot distinguish (with constant probability) between $\bv\sim \Dyes^{(\ell)}(w)$ and $\bu$ which is drawn according to any of the following distributions:
	\begin{enumerate}
		\item $\Dyes^{(\ell)}(w')$ for any $w'\neq w$.\label{itm:sharedsecet }
		\item $\Dno^{(\ell)}$. \label{itm:FarFromEncoding}
	\end{enumerate}

\end{lemma}
Note that Item~(\ref{itm:sharedsecet }) in the above follows immediately from Item~(\ref{itm:FarFromEncoding}). Additionally, the first item implies the shared secret property of the code ensemble $\calE ^{(\ell)}$. Furthermore, we remark that that above lemma implies a more stringent version of $\PCUSS$. In addition to the shared secret property, Item~(\ref{itm:FarFromEncoding}) implies that the ensemble $\calE ^{(\ell)}$ is indistinguishable from strings that are mostly far from any encoding  (i.e., drawn from $\Dno^{(\ell)}$).

\medskip
The proof of Lemma~\ref{lem:StrongerLB} follows by induction over $\ell$.
Before we continue, we introduce some useful lemmas that will be used in the proof.
\begin{lemma}\label{lem:minDist-E} For any $\ell\ge 0$ and $w,w'\in\zo^k$ for which $w\neq w'$ it holds that \[ \min_{(v,v')\in \calE^{(\ell)}(w)\times \calE^{(\ell)}(w')}\dist (v,v')= {\Theta\left(1/4^{\ell+1}\right) }\]
\end{lemma}
\begin{proof} The proof follows by induction over $\ell$. {The base case for  $\ell=0$ follows directly by the fact that the code from Definition~\ref{def:HardCode} has high distance, and in particular $\dist(\calE^{(0)}(w),\calE^{(0)}(w'))>1/10$.}
Assume that the lemma holds for $\ell-1$. Namely, for $w,w'\in\zo^{k'}$ for which $w\neq w'$ it holds that
\[\min_{(v,v')\in \calE^{(\ell-1)}(w)\times \calE^{(\ell-1)}(w')}\dist(v,v')=\Theta\left(1/4^\ell\right).\]
Let $\tilde{w},\tilde{w}'\in \zo^{k}$ be such that $\tilde{w}'\neq \tilde{w}$. Then we can write $(\tilde{v},\tilde{v}')\in \calE^{(\ell)}(\tilde{w})\times \calE^{(\ell)}(\tilde{w}')$ as
\begin{align*}
	\tilde{v}=\bigsqcup_{\beta\in \F\setminus H}\calE^{(\ell-1)}(\qu {g(\beta)})\qquad \text{and }\qquad 	\tilde{v}'=\bigsqcup_{\beta\in \F\setminus H}\calE^{(\ell-1)}(\qu {g'(\beta)}),
\end{align*}
for some $g,g'\in \calC_{\F}$ such that {$g|_H=\tilde{w}$ and $g'|_H=\tilde{w}'$}. By the fact that $g$ and $g'$ are degree $|\F |/2$ polynomials (which are not identical), we have that $g$ and $g'$ disagree on at least $|\F \setminus H|/4$ of the elements $\beta\in \F \setminus H$. By applying the induction hypothesis on the minimum distance between $\calE^{(\ell)}(\qu {g(\beta)})$ and  $\calE^{(\ell)}(\qu {g'(\beta)})$, for all $\beta$ such that $g(\beta)\neq g'(\beta)$, we have that \[\min_{(\tilde{v},\tilde{v}')\in \calE^{(\ell)}(\tilde w)\times \calE^{(\ell)}(\tilde w')}\dist (\tilde v,\tilde v')>\frac{1}{4}\cdot \Theta\left(\frac{1}{4^\ell}\right)=\Theta\left(1/4^{\ell+1}\right). \] \end{proof}
\begin{lemma} \label{lem:BaseDistance}For any $\ell\ge 0$, with probability at least $1-o(1)$, a string $v$ drawn from $\Dno^{(\ell)}$ satisfies $\dist(v,\calE^{(\ell)}(w))=\Theta\left(1/4^{\ell+1}\right)$ for all $w\in\{0,1\}^k$.
\end{lemma}
\begin{proof}  The proof follows by induction over $\ell$. {For $\ell=0$, fix some $w\in\zo^k$. Consider the size of a ball of relative radius $1/40$ around some  $v\in\calE ^{(0)}(w)$ in the space of all strings $\zo^{4k}$. The number of strings contained in this ball is at most
	\[\binom{4k}{k/10}\le (40e)^{k/10}=2^{{k/10}\cdot\log(40e)}.\]
Thus, the size of the set of strings which are at relative distance $1/40$ from any legal encoding of some word $w\in\zo^k$ is at most \[2^{3k}\cdot 2^{{k/10}\cdot\log(40e)}=o(2^{4k}).\]
This implies that with probability at least $1-o(1)$, a random string from $\zo^{4k}$ is $1/40$-far from $\calE ^{(0)}(w)$ for any $w\in\zo^k$.}

	For any $\ell>0$, consider $v'$ sampled according to $\Dno^{(\ell)}$. Then, $v'$ can be written as
	\[v'=\bigsqcup_{\beta\in \F\setminus H}\calE^{(\ell-1)}(\qu {\blambda(\beta)}),\] where $\blambda: \F \setminus H\to \F$ is a uniformly random function. On the other hand, each member $\tilde{v}$ of $\calP^{(\ell)}$ can be written as
	\[\tilde v=\bigsqcup_{\beta\in \F\setminus H}\calE^{(\ell-1)}(\qu {g(\beta)}),\] for some $g\in \calC_{\F}$ such that {$g|_H=w$} for some $w\in\zo^k$. Note that by Lemma~\ref{lem:minDist-E}, whenever $\blambda(\beta)\neq g(\beta)$, we have that the minimum distance between any $\tilde v\in \calE^{(\ell-1)}(\qu {g(\beta)})$ and $v' \in  \calE^{(\ell-1)}(\qu {\blambda(\beta)})$ is at least $\Theta(1/4^\ell)$. In addition, by Lemma~\ref{lem:rand-function-far-F}, we have that that with probability at least $1-o(1)$, a uniformly random function $\blambda :\F \to \F$ is $1/3$-far from any $g\in \calC_{\F}$. By the restrictions on $k$ in Definition~\ref{def:iteratedEnc}, which implies that $|H|\le|F|/c$, we can ensure (by the choice of $c$) that with probability at least $1 - o(1)$, that a uniformly random $\blambda:\F \setminus H\to \F $ is at least $1/4$-far from the restriction $g|_{\F \setminus H}$.
	This implies that for at least $|\F \setminus H|/4$ of the elements $\beta\in \F \setminus H$, we have that $\blambda(\beta)\neq g(\beta)$. Therefore, we have that $\dist(v',\calE ^{(\ell)}(w))=\frac{1}{4}\cdot \Theta\left(\frac{1}{4^\ell}\right)=\Theta\left(1/4^{\ell+1}\right)$ for all $w\in\zo^k$, and the proof is complete.
\end{proof}

\medskip

\begin{lemma} \label{lem:Base-id-distributed}  Fix any $\ell>0$, and suppose that for any $w'\in\zo^{k'}$, and for any set $Q'$ of at most $q^{(\ell-1)}_{\F ',k'}$ queries (where $\F '$ and $k'$ are picked according to the recursive definition of the level $\ell$-encoding, and for $q^{(0)}$ we substitute $k'$ for the nonexistent $\F'$) the restricted distributions $\Dyes^{(\ell-1)}(w')|_{Q'}$ and $\Dno^{(\ell-1)}|_{Q'}$ are identical. Then, for any $w\in \zo^k$, and any set $Q$ of at most $\frac{|\F \setminus H|}{10}\cdot q^{(\ell-1)}_{\F ',k'}$ queries, the restricted distributions $\Dyes^{(\ell)}(w)|_Q$ and $\Dno^{(\ell)}|_Q$ are identical.
\end{lemma}
\begin{proof} Let $Q\subset [m^{(\ell)}]$ be the set of queries, and fix a canonical ordering over the elements in $\F\setminus H$.  Let $\bv$ be an element drawn according to distribution $\Dyes^{(\ell)}(w)$, and let $\bv'$ be an element drawn according to distribution $\Dno^{(\ell)}$. The sampling process from $\Dyes^{(\ell)}(w)$ can be thought of as first drawing a uniformly random function $\bg\in \calC_{\F}$ such that {$\bg|_H=w$}, and for every $\beta\in\F \setminus H$, letting $\bv_\beta$ be a uniformly random element in $\calE^{(\ell-1)}(\qu {\bg(\beta)})$.

	\noindent For each $\beta \in \F\setminus H$ we set $Q_\beta= Q\cap \Dom(v_\beta)$, and define the set of \emph{big clusters}  $$I=\left\{ \beta \in \F \setminus H\ : \   |Q_\beta|\ge q^{(\ell-1)}_{\F ',k'} \right\}.$$
	Note that since $|Q|\le |\F \setminus H|\cdot q^{(\ell-1)}_{\F ',k'}/10 $, we have that $|I|\le |\F \setminus H|/10$.

	By the fact that $\bg$ is a uniformly random polynomial of degree  $|\F |/2 > |I|$, we have that $\bg|_{I}$ is distributed exactly as $\blambda|_{I}$ (both are a sequence of $|I|$ independent uniformly random values), which implies that $\bv|_{\bigcup_{j\in I}Q_j}$ is distributed exactly as $\bv'|_{\bigcup_{j\in I}Q_j}$.

	Next, let $\F \setminus (I \cup H)=\{i_1,\ldots,i_{|\F\setminus (I\cup H)|}\}$ be a subset ordered according to the canonical ordering over $\F$.  We proceed by showing that $\bv|_{\bigcup_{j \in I \cup \{i_1,...,i_t\}}Q_j}$ is distributed identically to $\bv'|_{\bigcup_{j \in I \cup \{i_1,...,i_t\}}Q_j}$ by induction over $t$.


	The base case ($t=0$) corresponds to the restriction over $\bigcup_{j \in I} Q_j$, which was already proven above. For the induction step, let $T=\{i_1,\ldots,i_{t-1}\}\subseteq \F \setminus (I\cup H)$ be an ordered subset that agrees with the canonical ordering on $\F$, and let $i_{t}\in \F \setminus(H \cup  T \cup I)$ be the successor of $i_{t-1}$ according to the ordering. We now prove that for each $x\in\{0,1\}^{m^{(\ell)}}$ for which $\bv|_{\bigcup_{j\in I\cup T}Q_j}$ has a positive probability of being equal to $x|_{\bigcup_{j\in I\cup T}Q_j}$, conditioned on the above event taking place (and its respective event for $v'$), $\bv|_{Q_{i_{t}}}$ is distributed exactly as $\bv'|_{Q_{i_{t}}}$.

	Observe that conditioned on the above event, $\bv|_{Q_{i_{t}}}$ is distributed exactly as a uniformly random element in $\calE ^{(\ell-1)}(\brho)$ for some $\brho\in\{0,1\}^{k'}$ (which follows some arbitrary distribution, possibly depending on $x|_{\bigcup_{j\in I\cup T}Q_j}$), while $\bv'|_{Q_{i_{t}}}$ is distributed exactly as a uniformly random element in $\calE ^{(\ell-1)}(\by)$ for a uniformly random $\by\in \zo^{k'}$. By the fact that $|Q_{i_{t}}|\le q^{(\ell-1)}_{\F ',k'}/10$, we can apply the induction hypothesis and conclude that $\bv|_{Q_{i_{t}}}$ is distributed exactly as $\bv'|_{Q_{i_{t}}}$, because by our hypothesis both are distributed identically to the corresponding restriction of $\Dno^{(\ell-1)}$, regardless of the values picked for $\brho$ and $\by$. This completes the induction step for $t$. The lemma follows by setting $t=|\F\setminus H\cup  I|$.
\end{proof}

\begin{lemma} \label{lem:iter-Queries}For any $\ell\ge 0$, $w\in\zo^k$ and any set of queries $Q\subset[m^{(\ell)}]$ such that $|Q|=O\left(\frac{m^{(\ell)}}{10^\ell}\right)$, the {restricted} distributions $\Dyes^{(\ell)}(w)|_Q$ and $\Dno^{(\ell)}|_Q$ are identically distributed.
\end{lemma}

\begin{proof} By induction on $\ell$. {For $\ell=0$ and any $w\in\zo^k$, by
	the fact that our base encoding $\calE ^{(0)}(w)$ is a high dual distance code,
we can select (say) $q^{(0)}=k/c$ (for some constant $c>0$), making the assertion of the lemma trivial.}

Assume that for any $w'\in\zo^{k'}$, and any set of queries $Q'$ of size at most $O(m^{(\ell-1)}/10^{\ell-1})$ the conditional distributions $\Dyes^{(\ell-1)}(w')|_{Q'}$ and $\Dno^{(\ell-1)}|_{Q'}$ are identically distributed. Then, by Lemma~\ref{lem:Base-id-distributed}, we have that for any $w\in \zo^k$ and any set of queries $Q$ of size at most $$O\left(\frac{|\F \setminus H|}{10^\ell}\cdot {m^{(\ell-1)}}\right),$$
the restricted distributions $\Dyes^{(\ell)}(w)|_Q$ and $\Dno^{(\ell)}|_Q$ are identically distributed. Note that by definition of the level $\ell$-encoding, {$m^{(\ell)}=|\F \setminus H|\cdot m^{(\ell-1)}$}, which implies the conclusion of the lemma.\end{proof}

\medskip

\begin{proofof}{Lemma~\ref{lem:StrongerLB}} Lemma~\ref{lem:StrongerLB} follows directly by combining Lemma~\ref{lem:ditsributed-exactly-indist}, and Lemma~\ref{lem:iter-Queries}.
\end{proofof}

{\medskip
Combining Lemma~\ref{lem:StrongerLB} with the definition of $\Spiel$-$\PCU$ (Definition~\ref{def:SpielPSU}) establishes that we have constructed a $\Spiel$-$\PCUSS$, which implies Theorem~\ref{thm:PCUSS-construction}.
\begin{corollary} Let $\F$ be a finite field and $k\in\N$ which satisfy the requirements in Definition~\ref{def:iteratedEnc}. Then, for every $\ell\ge 0$, the coding ensemble $\calE_{\F  ,k}^{(\ell)}:\zo^k\to 2^{\left(\zo^{m^{(\ell)}}\right)}$ has $q(\epsilon,\delta)$-query length-$t$ $\Spiel$-$\PCUSS$ for $t=O(m^{(\ell)}\polylogell m^{(\ell)})$.
\end{corollary}

}

\subsection{Handling arbitrary input lengths}
As mentioned in the beginning of this section, our construction of code ensembles relies on the fact that operations over a finite field $\mathrm{GF}(2^t)$ can be computed efficiently. In order to do so we need to have an irreducible polynomial of degree $t$ over $\mathrm{GF}(2)$, so that we have a representation $\mathrm{GF}(2^t)$. Given such a polynomial, operations over the field can be implemented in polylogarithmic time in the size of the field.
By~\cite{Gol08} (Appendix G), we know that for $t=2\cdot 3^r$ where $r\in \N$, we do have such a representation. However, the setting of $t$ restricts the sizes of the fields that we can work with, which will limit our input size length.

We show here how to extend our construction to a set of sizes that is ``log-dense''. For a global constant $c'$, our set of possible input sizes includes a member of $[m',c'm']$ for every $m'$. Moving from this set to the set of all possible input sizes now becomes a matter of straightforward padding.

 For any $n\in \N$, let $r$ be the smallest integer such that $n<2^{2\cdot3^r}$ and let $\F =\mathrm{GF}(2^{2\cdot 3^r})$.  We make our change only at the level-$\ell$ construction. First, we use $4d$ instead of $d$ in the calculation of the size of $\F'$. Then, instead of using $\F\setminus H$ as the domain for our input, we use $E\setminus H$, for any arbitrary set $E \subseteq \F$ of size  $n\ge \max\{4k,|\F|^{1/4},c_\ell\}$ that contains $H$. Then, for the level-$\ell$, instead of considering polynomials of degree $|\F|/2$, we consider polynomials of degree $|E|/2$. The rest of the construction follows the same lines as the one defined above. This way, all of our operations can be implemented in polylogarithmic time in $|E|$.

\section{Separations of testing models}

In this section we use Theorem~\ref{thm:PCUSS-construction} to prove a separation between the standard testing model, and both the tolerant and the erasure resilient testing models. Specifically, we prove the following.

\begin{theorem}[Restatement of Theorem~\ref{thm:Tol-Sep-Intro}] \label{thm:Tol-Sep-Notintro}For every constant $\ell\in \N$, there exist a property $\calQ^{(\ell)}$ and $\eps_1=\eps_1(\ell)\in (0,1)$ such that the following hold.
	\begin{enumerate}
		\item For every $\eps\in(0,1)$, the property $\calQ ^{(\ell)}$ can be $\eps$-tested using a number of queries depending only on $\eps$ (and $\ell$).
		\item For every $\eps_0\in (0,\eps_1)$, any $(\eps_0,\eps_1)$-tolerant tester for $\calQ^{(\ell)}$ needs to make $\Omega(N/10^\ell \cdot \polylogell N)$ many queries on inputs of length $N$.
	\end{enumerate}
\end{theorem}

\begin{theorem}[Restatement of Theorem~\ref{thm:ER-Sep-Intro}]\label{thm:ER-Sep-notintro} For every constant $\ell\in \N$, there exist a property $\calQ^{(\ell)}$ and $\eps_1=\eps_1(\ell)\in (0,1)$ such that the following hold.
	\begin{enumerate}
		\item For every $\eps\in(0,1)$, the property $\calQ^{(\ell)}$ can be $\eps$-tested using a number of queries depending only on $\eps$ (and $\ell$).
		\item For every $\eps\in (0,\eps_1)$ and any $\alpha=\Omega(1/\log^{(\ell)}N)$ satisfying $\eps +\alpha<1$, any $\alpha$-erasure resilient $\eps$-tester for $\calQ^{(\ell)}$ needs to make $\Omega(N/10^\ell \cdot \polylogell N)$ many queries on inputs of length $N$.
	\end{enumerate}
\end{theorem}

{In order to prove the separation we use the code ensemble $\calE_{\F ,k}^{(\ell)}$ where $k$ is set to $0$. Namely, we consider $\calE_{\F  ,0}(\emptyset)$. Note that in this case, the code ensemble becomes a property (i.e. a subset of the set of all possible strings).

{Next, we define the property that exhibits the separation between the standard testing model and both the tolerant testing model and the erasure resilient model. We prove Theorem~\ref{thm:Tol-Sep-Notintro} and mention the small difference between the proof of Theorem~\ref{thm:Tol-Sep-Notintro} and the proof of Theorem~\ref{thm:ER-Sep-notintro}.}

\begin{definition} Fix a finite field $\F$ and a constant integer $\ell\in \N$ and let $\eps(\ell)=\Theta(1/4^{\ell})$. Let $n\eqdef m^{(\ell)}_\F$, $z^{(\ell)}_{\F, 0}\le n\cdot\polylogell n$ denote the length of the proof for the $\PCUSS$ from Theorem~\ref{thm:PCUSS-construction}, and let $N=(\log^{(\ell)}n+1)\cdot z^{(\ell)}_{\F, 0}$. Let $\calQ^{(\ell)}\subseteq \zo^N$ be defined as follows. A string $x\in\zo^N$ satisfies $\calQ^{(\ell)}$ if the following hold.
	\begin{enumerate}
		\item The first $ z^{(\ell)}_{\F, 0}\cdot\log^{(\ell)}n$ bits of $x$ consist of $s= \frac{z^{(\ell)}_{\F, 0}\cdot\log^{(\ell)}n}{ n}$ copies of $y\in \calE_{\F,0}^{(\ell)}$.
		\item The remaining $z^{(\ell)}_{\F, 0}$ bits of $x$ consist of a proof string $\pi\in \zo^{z^{(\ell)}_{\F, 0}}$, for which the\\ \texttt{Verifier-Procedure}$_{\;\calE_{\F,0}^{(\ell)}}$ in Figure~\ref{fig:VerifierProcedureIterated} accepts $y$ given oracle access to $y$ and $\pi$.
	\end{enumerate}
\end{definition}

We first show that $\calQ^{(\ell)}$ can be tested using a constant number of queries in the standard testing model.

\begin{figure}[ht!]
	\begin{framed}
		\noindent \underline{$\texttt{Testing Algorithm for } \calQ^{(\ell)}$}
		\begin{flushleft}
			\noindent {\bf Input:}  Parameter
			$\eps\in (0,1)$, an oracle access to
			$x\in\zo^N$.
			\begin{enumerate}
				\item Set $s\eqdef \frac{z_{\F ,0}^{(\ell)}\cdot\log^{(\ell)}n}{ n}$.
				\item Repeat $4/\eps$ times:\label{step:LoopCalQ}
				\begin{enumerate}
					\item Sample $\bj\in [n]$ and $\bi\in[s]\setminus \{1\}$ uniformly at random.
					\item If $x_{\bj}\neq x_{(\bi-1)\cdot n+\bj}$, then \textbf{Reject}.
				\end{enumerate}
			\item Let $v=(x_1,\ldots,x_n)$, $\pi=(x_{z_{\F ,0}^{(\ell)}\cdot \log^{(\ell)}n+1},\ldots,x_{(\log^{(\ell)}n+1)z_{\F ,0}^{(\ell)}})$ and $\tau$ be the empty string.
			\item \label{item:run verifier} Run the $\PCU$ verifier for $\calE^{(\ell)}_{\F,0}$ with parameters $\eps/3$ and $\delta=2/3$ on $v$, using $\pi$ as the alleged proof for $v$, and $\tau$ as the alleged value for $v$.
			\item If the $\PCU$ verifier rejects, then \textbf{Reject}; otherwise \textbf{Accept}.
			\end{enumerate}

		\end{flushleft}\vskip -0.14in
	\end{framed}\vspace{-0.25cm}
	\caption{Description of $\texttt{Testing Algorithm for } \calQ^{(\ell)}$ .\vspace{-0.25cm}} \label{fig:QpropertyTester}
\end{figure}
	For Item~\ref{item:run verifier} in Figure \ref{fig:QpropertyTester}, recall that running the PCU verifier with parameter $\delta = 2/3$ actually involves running multiple instances of the verifier with smaller $\delta$, as discussed in Subsection \ref{subsec:prof_of_main_lemma}. 

\begin{lemma} \label{lem:Tol-Sep-Easy}The property $\calQ^{(\ell)}$ has a tester with query complexity depending only on $\eps$.
\end{lemma}

\begin{proof} We show that the algorithm described in Figure~\ref{fig:QpropertyTester} is a testing algorithm for $\calQ^{(\ell)}$.
We assume that $n$ is large enough so that $\log^{(\ell)}n> 6/\eps$.

	Assume that $x\in \calQ^{(\ell)}$. Then, there exists a string $y\in \calE_{\F,0}^{(\ell)}$, such that $x_1,\ldots,x_{z_{\F ,0}^{(\ell)}\log^{(\ell)}n}=(y)^s$ (where $(y)^s$ denotes the concatenation of $s$ copies of $y$), and $x_{z_{\F ,0}^{(\ell)}\cdot\log^{(\ell)}n +1},\ldots,x_{(\log^{(\ell)}+1)z_{\F ,0}^{(\ell)}}=\pi\in\zo^{z^{(\ell)}_{\F ,0}}$, where $\pi$ is a proof that makes the $\PCU$ verifier for $\calE_{\F,0}^{(\ell)}$  accept when given oracle access to $y$ and $\pi$. Therefore, the algorithm in Figure~\ref{fig:QpropertyTester} accepts $x$.

	\medskip
Next, assume that $x$ is $\eps$-far from $\calQ^{(\ell)}$, and let $y'=x_1,\ldots,x_n$. Note that if $x_1,\ldots,x_{z^{(\ell)}_{\F ,0}\cdot \log^{(\ell)}n}$ is $\eps/2$-far from being $(z')^s$, then the loop in Step~\ref{step:LoopCalQ} rejects $x$ with probability at least $2/3$, and we are done.
If $x_1,\ldots,x_{z^{(\ell)}_{\F, 0}\cdot \log^{(\ell)}n}$ is $\eps/2$-close to $(y')^s$, then $y'$ must be $\eps/3$-far from $\calE_{\F,0}^{(\ell)}$. To see this, assume toward a contradiction that $y'$ is $\eps/3$-close to $\calE_{\F,0}^{(\ell)}$. Then, by modifying at most $\frac{\eps\cdot z^{(\ell)}_{\F, 0}\cdot \log^{(\ell)}n }{2}$ bits, we can make $x_1,\ldots,x_{z^{(\ell)}_{\F, 0}\cdot \log^{(\ell)}n}$ equal to $(y')^s$. Since, by our assumption $y'$ is $\eps/3$-close to $\calE_{\F,0}^{(\ell)}$, we can further modify the string $(y')^s$ to $(\tilde{y})^s$, where $\tilde{y}\in\calE_{\F,0}^{(\ell)}$, by changing at most $\frac{\eps \cdot z^{(\ell)}_{\F, 0}\cdot \log^{(\ell)}n}{3}$ bits. Finally, by changing at most $z^{(\ell)}_{\F, 0}$ bits from $\pi$, we can get a proof string $\tilde\pi$ which will make the $\PCPP$ verifier  accept $\tilde{y}$. By our assumption that $6/\eps<\log^{(\ell)}n$, the total number of changes to the input string $x$ is at most
\[ \frac{\eps \cdot z^{(\ell)}_{\F, 0}\cdot \log^{(\ell)}n}{2}+\frac{\eps \cdot z^{(\ell)}_{\F, 0}\cdot \log^{(\ell)}n}{3}+z^{(\ell)}_{\F, 0}\le \eps\cdot (\log^{(\ell)}n+1)\cdot z^{(\ell)}_{\F, 0}=\eps N,  \]
which is a contradiction to the fact that $x$ is $\eps$-far from $\calE^{(\ell)}_{\F, 0}$.

Finally, having proved that $y'$ is $\eps /3$-far from $\calE_{\F,0}^{(\ell)}$, the $\PCU$ verifier for $\calE_{\F,0}^{(\ell)}$ (when called with parameters $\eps/3$ and $\delta=2/3$) rejects with probability at least $2/3$.
\end{proof}

\begin{lemma} \label{lem:Tol-Sep-Hard}For every constant $\ell\in\N$, there exists $\eps_1\eqdef \Theta(1/4^{\ell})$ such that for every $\eps_0<\eps_1$, any $(\eps_0,\eps_1)$-tolerant tester for $\calQ^{(\ell)}$ needs to make at least $\Omega\left(\frac{N}{10^\ell\cdot \polylogell N}\right)$ many queries.
\end{lemma}

\begin{proof}Fix some constant $\ell\in \N$. The proof follows by a reduction from $2\eps_1$-testing of $\calE^{(\ell)}_{\F,0}$. Given oracle access to a string $y\in \zo^n$ which we would like to $2\eps_1$-test for $\calE_{\F,0}^{(\ell)}$, we construct an input string $x\in\zo^N$ where $N=(\log^{(\ell)}n +1)\cdot z^{(\ell)}_{\F, 0}$ as follows.
	\[x\eqdef (y)^{\frac{z^{(\ell)}_{\F, 0}\cdot \log^{(\ell)}n}{ n}}\sqcup (0)^{z^{(\ell)}_{\F, 0}}.\]
That is, we concatenate $z^{(\ell)}_{\F, 0}\cdot \log^{(\ell)}n/ n$ copies of $y$, and set the last $z^{(\ell)}_{\F, 0}$ bits to $0$. Note that a single query to the new input string $x$ can be simulated using at most one query to the string $y$.

If $y\in\calE_{\F,0}^{(\ell)}$, then for large enough $n$ we have that $x$ is $\eps_0$-close to $\calQ^{(\ell)}$, since the last $z^{(\ell)}_{\F, 0}$ bits that are set to $0$ are less than an $\eps_0$-fraction of the input length.

On the other hand, if $\dist(x,\calE_{\F,0}^{(\ell)})>2\eps_1$, since each copy of $y$ in $x$ is $2\eps_1$-far from $\calE_{\F,0}^{(\ell)}$, then $x$ is $\frac{2\eps_1\cdot\log^{(\ell)}n}{\log^{(\ell)}n+1}$-far from $\calQ^{(\ell)}$ (note that $\frac{\log^{(\ell)}n}{\log^{(\ell)}n+1}>1/2$).
Therefore, an $(\eps_0,\eps_1)$-tolerant tester for $\calQ^{(\ell)}$ would imply an $2\eps_1$-tester for $\calE_{\F,0}^{(\ell)}$ with the same query complexity. By Lemma~\ref{lem:StrongerLB}, since for some $\eps_1=\Theta(1/4^\ell)$, every $2\eps_1$-tester for $\calE_{\F,0}^{(\ell)}$ requires $\Omega(n/10^{\ell})$ queries on inputs of length $n$, any $(\eps_0,\eps_1)$-tolerant tester for $\calQ^{(\ell)}$ requires to make $\Omega\left(\frac{N}{10^\ell\cdot  \polylogell N}\right)$ many queries.
\end{proof}

\medskip

\begin{proofof}{Theorem~\ref{thm:Tol-Sep-Notintro}}
The proof follows by combining Lemma~\ref{lem:Tol-Sep-Easy} and Lemma~\ref{lem:Tol-Sep-Hard}.
\end{proofof}

\medskip
\begin{proofof}{Theorem~\ref{thm:ER-Sep-notintro}} The proof of Theorem~\ref{thm:ER-Sep-notintro} is almost identical to the proof of Theorem~\ref{thm:Tol-Sep-Notintro}. The only difference is that we replace Lemma~\ref{lem:Tol-Sep-Hard} with a counterpart for erasure resilient testing, where instead of setting the last $z^{(\ell)}_{\F ,0}$ bits of $x$ to $(0)^{z^{(\ell)}_{\F ,0}}$, we use $(\bot)^{z^{(\ell)}_{\F ,0}}$, noting that the relative size of this part of the input is $1/(s+1)=\Theta(1/\log^{(\ell)}(N))$.
\end{proofof}

\begin{flushleft}
	\bibliographystyle{alpha}
	\bibliography{Levi}
\end{flushleft}
\end{document}